\documentclass{LMCS}

\def\doi{8(4:18)2012}
\lmcsheading%
{\doi}
{1--36}
{}
{}
{Jun.~\phantom01, 2012}
{Nov.~27, 2012}
{}

\usepackage{verbatim}
\usepackage{graphicx}
\usepackage{epstopdf}
\usepackage{bussproofs}
\usepackage{color} 
\usepackage{dsfont}
\usepackage{stmaryrd}
\usepackage{amsthm}
\usepackage{mathtools}
\usepackage{amssymb}
\usepackage{hyperref}
\usepackage{pstricks,pst-node,pst-tree}
\usepackage{enumerate}
\usepackage[all]{xy}
\usepackage{subfigure}
\usepackage{url}

\newcommand{\bydef}{ \stackrel{\mathrm{def}}{=} }

\newcommand{\freccia}[1]{\stackrel{#1}{\longrightarrow}}

\newcommand{\muprod}{\mu^{\odot}}

\newcommand{\diam}[1]{ \langle #1\rangle }
\newcommand{\quadrato}[1]{\left[ #1 \right]}

\newcommand{\lts}{\mathcal{L}}
\newcommand{\uno}{1}
\newcommand{\due}{2}
\newcommand{\supp}{ \textnormal{supp} }
\newcommand{\first}{ \textnormal{first} }
\newcommand{\last}{ \textnormal{last} }
\newcommand{\rootnode}{ \textnormal{root} }

\newcommand{\expected}{\mathbf{E}}
\newcommand{\martin}{$\textnormal{MA}_{\aleph_1}$}

\newcommand{\arena}{\mathcal{A}}

\newcommand{\val}{\mathsf{VAL}}
\newcommand{\game}{\mathcal{G}}
\newcommand{\pr}{\Omega}
\newcommand{\pl}{\mathsf{Pl}}

\newcommand{\bp}{\mathcal{BP}}

\newcommand{\lfp}{\mathrm{lfp}}
\newcommand{\gfp}{\mathrm{gfp}}

\newcommand{\sem}[1] {  \llbracket #1 \rrbracket  }  
\newcommand{\gsem}[1]{\llparenthesis{\,#1\,}\rrparenthesis}

\newcommand\node[1]{*+[o]{#1}}
\newcommand\nodeC[1]{*+[o][F]{#1}}

\newcommand\addLabelUL[1]{\ar@{}[]+UR|(1){~\makebox[0pt][l]{$\mathbf{#1}$}}}

\newcommand\addLabelUR[1]{\ar@{}[]+UR|(1){~\makebox[0pt][l]{$\mathbf{#1}$}}}

\newcommand\addLabelDL[1]{\ar@{}[]+UR|(1){~\makebox[0pt][l]{$\mathbf{#1}$}}}

\newcommand\addLabelDR[1]{\ar@{}[]+UR|(1){~\makebox[0pt][l]{$\mathbf{#1}$}}}

\newcommand\addDMD[2]{
	\ar@{-}[]+<#1pt,0pt>;[]+<0pt,#2pt>
	\ar@{-}[]+<0pt,#2pt>;[]+<-#1pt,0pt>
	\ar@{-}[]+<-#1pt,0pt>;[]+<0pt,-#2pt>
	\ar@{-}[]+<0pt,-#2pt>;[]+<#1pt,0pt>
}

\begin{document}

\title[Probabilistic Modal $\mu$-Calculus with Independent Product]{Probabilistic Modal $\mu$-Calculus\\ with Independent Product\rsuper*}
\titlecomment{{\lsuper*} A conference version of this paper, not including proofs, appeared as \cite{MIO11}.}

\author[M.~Mio]{Matteo Mio} 
\address{LIX, Ecole Polytechnique}
\email{mio@lix.polytechnique.fr} 
\thanks{
This research was partially supported by PhD studentships from LFCS and the IGS at the School of Informatics, University of Edinburgh, and by EPSRC research grant EP-F042043-1.
It was completed during the tenure of an ERCIM ``Alain Bensoussan'' Fellowship, supported by the Marie Curie Co-funding of Regional, National and International Programmes (COFUND) of the European Commission. 
}



\maketitle

\begin{abstract}
The \emph{probabilistic} modal $\mu$-calculus is a fixed-point logic designed for expressing properties of probabilistic labeled transition systems (PLTS's). Two equivalent semantics have been studied for this logic, both assigning to each state a value in the interval $[0,1]$ representing the probability   that the property expressed by the formula holds at the state. One semantics is \emph{denotational} and the other is a \emph{game semantics}, specified in terms of two-player stochastic parity games.

A shortcoming of the probabilistic modal $\mu$-calculus is the lack of expressiveness required to encode other important temporal logics for PLTS's such as Probabilistic Computation Tree Logic (PCTL). To address this limitation we extend the logic with a new pair of operators: independent product and coproduct. The resulting logic, called \emph{probabilistic modal $\mu$-calculus with independent product}, can encode many properties of interest and subsumes the qualitative fragment of PCTL. 

The main contribution of this paper is the definition of an appropriate game semantics for this extended probabilistic $\mu$-calculus. This relies on the definition of a new class of games which generalize standard two-player stochastic (parity) games by allowing a play to be split into concurrent subplays, each continuing their evolution independently. Our main technical result is the equivalence of the two semantics. The proof is carried out in ZFC set theory extended with Martin's Axiom at an uncountable cardinal.
\end{abstract}


\section{Introduction}
The modal $\mu$-calculus (L${\mu}$) \cite{Kozen83,Stirling96,BS2001} is a very expressive logic, for expressing properties of labeled transition systems (LTS's), obtained by extending classical propositional modal logic with least and greatest fixed point operators. In the last decade, a lot of research has focused on the study of reactive systems that exhibit some kind of probabilistic behavior, and logics for expressing their properties.  Probabilistic labeled transition systems (PLTS's) \cite{S95} are a natural generalization of standard LTS's to the probabilistic scenario, as they allow both non-deterministic and (countable) probabilistic choices. A state $s$ in a PLTS can evolve by non-deterministically choosing one of the \emph{accessible} probability distributions $d$ (over states) and then continuing its execution from the state $s^{\prime}$ with probability $d(s^{\prime})$. This combination of non-deterministic choices immediately followed by probabilistic ones, allows the modeling of concurrency, non-determinism and probabilistic behaviors in a natural way. PTLS's can be visualized using  graphs labeled with probabilities in a natural way \cite{HP2000,KNPV2009,Bartels02}. For example the PLTS depicted in Figure \ref{figura_intro_plts_1} models a system with two states $p$ and $q$. At the state $q$ no action can be performed. At the state $p$ the system can evolve non-deterministically either  to the state $q$ with probability $1$ (when the transition $p\freccia{a}d_{2}$ is chosen) or to the state $p$ with probability $\frac{1}{3}$ and to the state $q$ and with probability $\frac{2}{3}$ (when the transition $p\freccia{a}d_{1}$ is chosen). 
\begin{figure}[h!]
\begin{center}
$$
\SelectTips{cm}{}
	\xymatrix @=20pt {
		\nodeC{p} \ar@{->}[rr]^{a}  \ar@{->}[drr]^{a} & &  \nodeC{d_{1}} \ar@{.>}@/^10pt/[rr]^{\frac{2}{3}} \ar@{.>}@/_10pt/[ll]_{\frac{1}{3}} & &  \nodeC{q}  	\\
		& &   \nodeC{d_{2}} \ar@{.>}@/^0pt/[urr]^{1}\\
		}
$$
\end{center}
\caption{Example of a PLTS}\label{figura_intro_plts_1}
\end{figure}
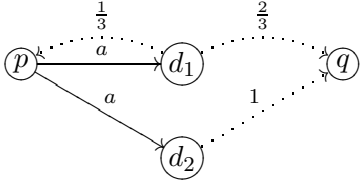

The probabilistic modal $\mu$-calculus (pL${\mu})$, introduced in \cite{MM97,HM96,AM04}, is a generalization of L${\mu}$ designed for expressing properties of PLTS's. 
This logic was originally named the \emph{quantitative} $\mu$-calculus, but since other $\mu$-calculus-like logics, designed for expressing properties of non-probabilistic systems, have been given the same name (see, e.g., \cite{FGK2010}), we adopt the \emph{probabilistic} adjective. The syntax of the logic pL$\mu$ coincides with that of the standard $\mu$-calculus. The denotational semantics  of pL${\mu}$ \cite{MM97,AM04} generalizes that of L$\mu$, by interpreting every formula $F$ as a map $\sem{F}\!:\!P\!\rightarrow\! [0,1]$, which assigns to each state $p$ a \emph{degree of truth}. In \cite{MM07}, the authors introduce an alternative semantics for the logic pL${\mu}$. This semantics, given in term of two-player stochastic parity games, is a natural generalization of the two-player (non stochastic) game semantics for the logic L$\mu$ \cite{Stirling96}. As in L$\mu$ games, the two players play a game starting from a
   configuration $\langle p, F\rangle$, where the objective for Player $1$ is
   to produce a path of configurations along which the outermost
   fixed point variable $X$ unfolded infinitely often is bound by a greatest
   fixed point in $F$. On a configuration of the form $\langle p,G_{1}\vee G_{2}\rangle$,
   Player $1$ chooses one of the disjuncts $G_{i}$, $i\!\in\!\{1,2\}$, by moving to the next
   configuration $\langle p, G_{i}\rangle$. On a configuration $\langle p,G_{1}\wedge G_{2}\rangle$,
   Player $2$ chooses a conjunct $G_{i}$ and moves to $\langle p, G_{i}\rangle$. On a configuration
   $\langle p, \mu X.G\rangle$ or  $\langle p, \nu X.G\rangle$ the game evolves to the configuration
   $\langle p, G\rangle$, after which, from any subsequent configuration $\langle q, X\rangle$ the game
   again evolves to $\langle q, G\rangle$. On configurations $\langle p, \diam{a}G\rangle$ and $\langle p, \quadrato{a}G\rangle$,
   Player $1$ and $2$ respectively choose a transition ${p}\!\freccia{a}\!{d}$ in
   the PLTS and move the game to $\langle d, G\rangle$. Here $d$ is a
   probability distribution (this is the key difference between
   pL$\mu$ and L$\mu$ games)  and the configuration $\langle d,G\rangle$ belongs to Nature, the probabilistic
   agent of the game, who moves on to the next configuration $\langle q,G\rangle$ with
   probability $d(q)$.
This game semantics allows one to interpret formulas as expressing, for each state $p$, the (limit) probability of a \emph{property}, specified by the formula, holding at the state $p$. In \cite{MM07}, the equivalence of the denotational and game semantics for pL$\mu$  on all finite models, was proven. The result was recently  extended to arbitrary models by the present author  \cite{MIO10}.

Having a complementary game semantics for the logic pL$\mu$ is of great conceptual importance. In the \emph{quantitative} approach to probabilistic temporal logics,  the truth value associated with a formula at a given state is supposed to represent the probability that the property expressed by the formula holds at the state. Since the connectives of pL$\mu$ can be given other meaningful denotational interpretations different from those considered in \cite{MM97,AM04} (see, e.g., \cite{HM96}), one naturally seeks an alternative description for the properties associated with formulas going beyond the mere denotational interpretation. The game semantics for pL$\mu$ provides such an  \emph{operational} interpretation in terms of the interactions between the controller (Player $1$) and a hostile environment (Player $2$) in the context of the stochastic choices occurring in the PLTS (Nature). 

However, a shortcoming of the probabilistic $\mu$-calculus is the lack of expressiveness required to encode other important temporal logics for PLTS's, such as Probabilistic Computation Tree Logic (PCTL) of \cite{BA1995}. To address this limitation,  we consider an extension of the logic pL$\mu$ obtained by adding to the syntax of the logic a second conjunction operator ($\cdot$) called \emph{product} and its De Morgan dual operator called \emph{coproduct} ($\odot$). We call this extension the \emph{probabilistic modal $\mu$-calculus with independent product}, or just pL$\muprod$. The denotational semantics of the product operator is defined as $\sem{F\cdot G}(p)\!=\!\sem{F}(p)\cdot\sem{G}(p)$, where the product symbol in the right hand side is multiplication on reals. The denotational semantics of the coproduct is defined, by De Morgan duality, as $\sem{F\odot G}(p)\!=\!1-\big((1-\sem{F}(p))\cdot(1-\sem{G}(p))\big)$. These operators were already considered in \cite{HM96} as a possible generalization of standard boolean conjunction and disjunction to the lattice $[0,1]$. Our logic pL$\muprod$ is novel in containing both ordinary conjunctions and disjunctions ($\wedge$ and $\vee$) and independent products and coproducts ($\cdot$ and $\odot$). While giving a denotational semantics to pL$\muprod$ is straightforward, the major task we undertake in this paper is to extend the game semantics of \cite{MM07} to the new connectives. The conceptual importance of this kind of result has already been outlined above (see also \cite[\S 1]{MioThesis} for an extensive exposition). The game semantics implements the intuition that $H_{1} \cdot H_{2}$ expresses the probability that $H_{1}$ and $H_{2}$ both hold if verified \emph{independently} of each other.

To capture formally this intuition we introduce a game semantics for the logic pL$\muprod$ in which independent execution of many instances of the game is allowed. Our games build on those for pL$\mu$ outlined above. Novelty arises in the game interpretation of the game states $\langle p, H_{1}\!\cdot\! H_{2}\rangle$ and $\langle p, H_{1}\odot H_{2}\rangle$. When during the execution of the game one of this kind of nodes is reached, the game is split into two concurrent and  independent subgames continuing their executions from the states $\langle p, H_{1}\rangle$ and $\langle p,H_{2}\rangle$ respectively. The difference between the game-interpretation of product and coproduct operators  is that on a product configuration $\langle p, H_{1}\cdot H_{2}\rangle$, Player $1$ has to win  in both generated sub-games, while on a coproduct configuration $\langle p, H_{1}\odot H_{2}\rangle$, Player $1$ needs to win just one of the two generated sub-games.
 
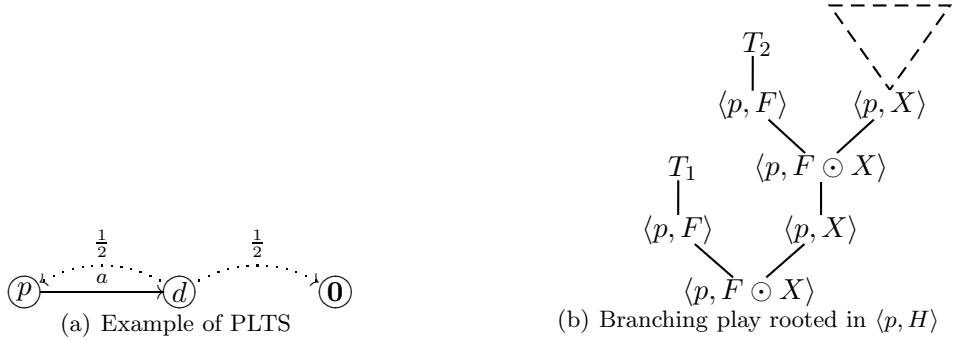
\begin{figure}[t]\label{figura_intro}
\centering
\subfigure[Example of PLTS]{\label{esempio_plts}
$$
\SelectTips{cm}{}
	\xymatrix @=20pt {
		\nodeC{p} \ar@{->}[rr]^{a} & &  \nodeC{d} \ar@{.>}@/^10pt/[rr]^{\frac{1}{2}} \ar@{.>}@/_10pt/[ll]_{\frac{1}{2}} & &  \nodeC{\mathbf{0}}  	}
$$
}\qquad\qquad\qquad
\subfigure[Branching play rooted in $\langle p, H\rangle$]{\label{intro_fig_play}
$\ \ \ \ \ \ \ \ \ $
\pstree[ treemode=U,levelsep=5ex ]{\Tr{$\langle p, F \odot X\rangle$}}{
 \pstree{\Tr{$\langle p, F\rangle$}}{
  		\TR{ $T_{1} $}
 }
 \pstree{\Tr{$\langle p, X\rangle$}}{
 	\pstree[ treemode=U,levelsep=5ex ]{\Tr{$\langle p, F \odot X\rangle$}}{
	 \pstree{\Tr{$\langle p, F\rangle$}}{
  		\TR{ $T_{2} $}
 }
 	\pstree[ treemode=U,levelsep=8ex,linestyle=dashed ]{\Tr{$\langle p, X \rangle$}}{
        \pstree[linestyle=none,arrows=-,levelsep=2ex]{\Tfan[fansize=10ex]}{\TR{ $\ $}}
	}
	}
}
}
}\caption{Illustrative example}
\end{figure}

To illustrate the main ideas, let us consider the PLTS of figure \ref{esempio_plts} and the pL$\mu$ formula $F\!=\!\diam{a}\diam{a}t\!t$ which asserts the possibility of performing two consecutive $a$-steps. The probability of $F$ being satisfied at $p$ is $\frac{1}{2}$,  since after the first $a$-step, the state $\bold{0}$ is reached with probability $\frac{1}{2}$ and no further $a$-step is possible. Let us consider the pL$\muprod$ formula $H\!=\! \mu X. (F\odot X)$. Figure \ref{intro_fig_play} depicts a play in the game starting
 from the configuration $\langle p, H\rangle$ (fixed-point unfolding steps are omitted).
 The branching points represent places where coproduct is the
 main connective, and each $T_{i}$ represents play in one of the
 independent subgames for $\langle p, F\rangle$ thereupon generated.  We call such a tree, describing play on all independent subgames, a
 \emph{branching play}. Since all branches are coproducts, and the
 fixpoint is a least fixpoint, the objective for Player $1$ is to win
 at least one of the games $T_i$. Since the probability of winning a
 particular game $T_i$ is $\frac{1}{2}$, and there are infinitely many
 independent such games, the probability of Player $1$  winning the whole
 game $H$ is $1$. Therefore the game semantics assigns $H$ at $p$
 the value $1$.

The above example illustrates an interesting application of the new operators, namely the  possibility of encoding the \emph{qualitative} probabilistic modalities $\mathbb{P}_{>0}F$ ($F$ holds with probability greater than zero) and $\mathbb{P}_{=1}F$ ($F$ holds with probability $1$), which are equivalent to the pL$\muprod$ formulas $\mu X.(F\odot X)$ and $\nu X.(F\cdot X)$ respectively. 
Other useful properties can be expressed by using these probabilistic modalities in the scope of fixed point operators. Some interesting formulas include $\mu X.\big( \diam{a}X \vee (\mathbb{P}_{=1}H)\big)$, $\nu X.\big(\mathbb{P}_{>0}\diam{a}X\big)$ and $\mathbb{P}_{>0}\big(\nu X.\diam{a}X\big)$. The first formula assigns to a state $p$ the probability of eventually reaching, by means of a sequence of $a$-steps, a state in which $H$ holds with probability $1$. The second, interpreted on a state $p$, has value $1$ if there exists an infinite sequence of possible (in the sense of having probabilty greater than $0$) $a$-steps starting from $p$, and $0$ otherwise. The third formula, expresses a stronger property, namely it assigns to a state $p$ value $1$ if the probability of making (starting from $p$) an infinite sequence of $a$-steps  is greater than $0$, and value $0$ otherwise. As a matter of fact the logic pL$\mu^{\odot}$ is expressive enough to encode the qualitative fragment of PCTL.

Formalizing the pL$\muprod$ games outlined above is a surprisingly
 technical undertaking. To account for the \emph{branching plays} that arise,
 we introduce a general notion of \emph{tree game} which is of
 interest in its own right. Tree games generalize two-player stochastic
 games, and are powerful enough to encode certain classes of games of
 imperfect information such as Blackwell 
 games \cite{Martin98}. A further level of difficulty arises in expressing
 when a branching play in a pL$\muprod$ game is considered an objective
 for Player $1$. This is delicate because branching plays
 can contain infinitely many interleaved occurrences of product and
 coproduct operations (so our simple explanation of such nodes
 above does not suffice). To account for this, branching plays
 are themselves considered as ordinary two-player (parity) games with
 coproduct nodes as Player $1$ nodes, and product nodes as Player $2$ nodes.
 Player $1$'s goal in the \emph{outer} pL$\muprod$ game is to produce a
 branching play for which, when itself considered as a game, the
 \emph{inner} game, they have a winning strategy. To formalize the class of tree games whose objective is specified by means of \emph{inner} games, we introduce the notion of \emph{two-player stochastic meta-parity game}.

Our main technical result is the equivalence of the denotational semantics and the game semantics for the logic pL$\muprod$. As in \cite{MIO10} the proof of  equivalence of the two semantics is based on the \emph{unfolding method} of \cite{FGK2010}. However there are significant complications, notably,
the transfinite inductive characterization of the set of winning branching plays in a given pL$\muprod$ game (section \ref{technical_section}) and the lack of denotational continuity on the free variables taken care by the game-theoretic notion of \emph{robust} Markov branching play (Section \ref{robust_markov_plays_section}). Moreover, because of the complexity of the objectives described by means of \emph{inner games}, the proof is carried out in ZFC set theory extended with  $\textrm{MA}_{\aleph_1}$  (Martin's Axiom at $\aleph_1$) and therefore our result is at least consistent with ZFC. We leave open the question of whether our result is provable in ZFC alone; we do not  know if this is possible even restricting the equivalence problem to finite models.

The rest of the paper is organized as follows. In Section \ref{basic} we discuss the required mathematical background.   In Section \ref{syntax} we define the syntax and the denotational semantics of the logic pL$\muprod$.
In Section \ref{tree_games_section} the class of stochastic tree games, and its sub-class given by two-player stochastic meta-parity games, are introduced in detail. In Section \ref{model_checking_games} the game semantics of the logic pL$\muprod$ is defined in terms of two-player stochastic meta-parity games. In Section \ref{technical_section} we provide a transfinite inductive characterization of the winning set of the game associated with a formula $\mu X.F$. In Section \ref{robust_markov_plays_section} we introduce the technical notion of robust Markov branching play.
In Section \ref{proof_section} we prove the equivalence of the two semantics, our main result. Conclusions and directions for future research are presented in Section \ref{conclusion_section}. 


\section{Mathematical Background}\label{basic}
\begin{defi}A (discrete) probability distribution $d$ over a set $X$ is a function $d\!:\!X\!\rightarrow\![0,1]$ such that $\sum_{x\in X}d(x)\!=\!1$. The \emph{support} of  $d$, denoted by $\supp(d)$, is defined as the (necessarily countable) set $\{x\!\in\! X\ | \ d(x)\!>\!0\}$. We denote with $\mathcal{D}(X)$ the set of probability distributions over $X$. We denote with $\delta_{x}$, for $x\!\in\!X$, the unique probability distribution such that $\supp(\delta_{x})\!=\!\{x\}$.
\end{defi}
\begin{defi}[PLTS \cite{S95}]\label{PLTS} Given a countable set $L$ of labels, a \emph{Probabilistic Labeled Transition System} is a pair $\lts\!=\!\langle P, \{ \freccia{a}\}_{a\in L}\rangle$, where $P$ is a \emph{countable} set of states and ${\freccia{a}}\! \subseteq\! P\times \mathcal{D}(P)$, for every $a\in L$. In this paper we restrict our attention to those PLTS's such that for every $p\!\in\! P$ and every $a\!\in\! L$, the set $\{d\ | \ p\freccia{a}d\}$ is countable. We refer to the countable set $\bigcup_{a\in L}\bigcup_{p\in P}\{d\ | \ p\freccia{a}d\}$, denoted by $\mathcal{D}(\lts)$, as the set of probability distributions of the PLTS. We say that a PLTS $\lts$ is \emph{not-probabilisitc}, or just a LTS, if every probability distribution $d\!\in\!\mathcal{D}(\lts)$ is of the form $\delta_{p}$, for some $p\!\in\!P$.
\end{defi}
Given a set $X$, we denote with $2^{X}$ the set of all subsets $Y\!\subseteq\! X$. Given a complete lattice $(L,\leq)$, we denote with $\bigsqcup\!:2^{L}\!\rightarrow\! L$ and $\bigsqcap\!:2^{L}\!\rightarrow\! L$ the operations of join and meet respectively. A function $f\!:\!L\!\rightarrow\!L$ is \emph{monotone} if $x\!\leq\! y$ implies $f(x)\!\leq\! f(y)$, for every $x,y\!\in\!L$. The set of fixed points of any monotone function $f\!:\!L\rightarrow\!L$, ordered by $\leq$, is  a non-empty complete lattice \cite{Tarski1955}. We denote with $\lfp(f)$ and $\gfp(f)$ the least and the greatest fixed points of $f$, respectively. 

\begin{thm}[Knaster--Tarski \cite{Tarski1955}]
Let $(L,\leq)$ be a complete lattice and $f\!:\!L\!\rightarrow\!L$ a monotone function. The following equalities hold:
\begin{enumerate}[\em(1)]
\item $\lfp(f) =\bigsqcup_{\alpha} f^{\alpha}$, where $f^{\alpha}=\bigsqcup_{\beta<\alpha}f(f^{\beta})$,
\item $\gfp(f) = \bigsqcap_{\alpha} f_{\alpha}$, where $f_{\alpha}=\bigsqcap_{\beta<\alpha}f(f_{\beta})$,
\end{enumerate}
where the greek letters $\alpha$ and $\beta$ range over ordinals.
\end{thm}

The closed real interval $[0,1]$, with its standard order $\leq$, is a (distributive) complete lattice \cite{MioThesis} with $\bigsqcup X\!=\! \sup  X $ and $\bigsqcap X\!=\! \inf X$, for every $X\!\subseteq\![0,1]$. When $X\!=\! \{x,y\}$ we simply have $x\sqcup y=\max\{x,y\}$ and $x\sqcap y=\min\{x,y\}$. The involutive map $x\mapsto 1-x$ is clearly order-reversing, thus the structure $([0,1], \leq, \lambda x.1-x)$ constitutes a complete De Morgan algebra. We shall consider the operation of product ($\cdot$), i.e., standard multiplication on $[0,1]$, and its De Morgan dual, called coproduct ($\odot$), defined as: $x\odot y\!\bydef\! 1-\big((1-x)\cdot (1-y)\big)$. Equivalently, $x\odot y\!=\! x+y-xy$. Both operations are commutative, associative, monotone and extend uniquely to infinitary operations of type $\prod,\coprod:[0,1]^{\mathbb{N}}\!\rightarrow[0,1]$ as follows:
\begin{center}
$\displaystyle \prod\{x_{n}\}_{n\in\mathbb{N}} = \bigsqcap_{m\in\mathbb{N}} x_{0}\cdot \ldots \cdot x_{m} \ \ \ \ $ and $\displaystyle \ \ \ \ \coprod\{x_{n}\}_{n\in\mathbb{N}} = \bigsqcup_{m\in\mathbb{N}} x_{0}\odot \ldots \odot x_{m}$.
\end{center}
We shall make use of the following fast growing function on the natural numbers to approximate infinite (co)products.
\begin{defi}\label{approx_function}
The function $\#\!:\!\mathbb{N}\!\rightarrow\!\mathbb{N}$ is defined as $\#(n)\bydef 2^{2^{n}+1}$.
\end{defi}
\begin{prop}
The following inequalities hold for every $\varepsilon\!\in\!(0,1]$ and $\vec{x}_{n}\!\in\![0,1]^{\mathbb{N}}$:
\begin{center}
$\displaystyle \prod_{n\in\mathbb{N}}\big( x_{n} + \frac{\varepsilon}{\#(n)}\big) \leq \big(\prod_{n\in\mathbb{N}} x_{n}\big) + \varepsilon \ \ \ \textnormal{and} \ \ \ \prod_{n\in\mathbb{N}}\big( x_{n} - \frac{\varepsilon}{\#(n)}\big) \geq \big(\prod_{n\in\mathbb{N}} x_{n}\big) - \varepsilon$
\end{center}
and, similarly, if infinitary products are replaced by infinitary coproducts.
\end{prop}

\begin{proof}
Both inequalities are easily proved by routine methods. We refer to Lemma 2.2.10 of \cite{MioThesis} for a detailed proof.
\end{proof}

In the following we assume standard notions of basic topology and basic measure theory. We refer to \cite{Kechris} and \cite{Tao_measuretheory} as  standard references to  these topics. The author's PhD thesis \cite[Chapter $2$]{MioThesis} provides a succinct introduction to the necessary material. 
 The topological spaces we consider will always be $0$-dimensional Polish spaces.  

Adopting standard notation, we denote with $\mbox{\boldmath$\Gamma$}^{0}_{\alpha}$, for $\Gamma\!\in\!\{\Sigma,\Pi,\Delta \}$ and $\alpha$ a countable ordinal greater than $0$,  the corresponding class of subsets (of a Polish space) in the Borel hierarchy. In particular, $\mbox{\boldmath$\Delta$}^{0}_{1}$, $\mbox{\boldmath$\Sigma$}^{0}_{1}$ and $\mbox{\boldmath$\Pi$}^{0}_{1}$ denote the collections of clopen, open and closed sets respectively. 

\begin{defi}\label{projective_symbol_list}
For each $n\!\geq\!1$ the \emph{projective classes}  $\mbox{\boldmath$\Sigma$}^{1}_{n}(X)$,  $\mbox{\boldmath$\Pi$}^{1}_{n}(X)$,  $\mbox{\boldmath$\Delta$}^{1}_{n}(X)$ of subsets of a Polish space $X$ are defined as follows:
\begin{center}
\begin{tabular}{l l l}
 $\mbox{\boldmath$\Sigma$}^{1}_{n+1}$ & $=$ & $\Big\{ A\!\subseteq\! X \ | \  A\!=\! \{ x \ | \ \exists y. (x,y)\!\in\!B\} \textit{ for some } B\!\in\! \mbox{\boldmath$\Pi$}^{1}_{n}(X\times Y) \Big\}$\\
 $\mbox{\boldmath$\Pi$}^{1}_{n+1}$& $=$ & $\{ A\!\subseteq\! X \ | \ (X\setminus A)\!\in\!  \mbox{\boldmath$\Sigma$}^{1}_{n+1}(X)\}$\\
  $\mbox{\boldmath$\Delta$}^{1}_{n}$ & $=$ &  $\mbox{\boldmath$\Sigma$}^{1}_{n} \cap  \mbox{\boldmath$\Pi$}^{1}_{n}$\\
\end{tabular}
\end{center}
where $Y$ is a Polish space, and $X\times Y$ is equipped with the product topology.
It is a well-known result of Susin (see, e.g.,  Theorem 14.11 in \cite{Kechris}) that  $\mbox{\boldmath$\Delta$}^{1}_{1}$ is the collection of Borel sets. Thus the class of projective sets is defined, starting from the Borel sets, by iterating the operations of projection and complementation.
\end{defi}

\begin{defi}
Given a Polish space $X$ and a subset $A\!\subseteq\! X$, we say that $A$ is \emph{universally measurable} if it is $\mu$-measurable for every Borel probability measure on $X$. It follows that the collection of universally measurable subsets of $X$, denoted by $\textnormal{UM}(X)$, forms a $\sigma$-algebra. Given a measurable space $(Y,\mathcal{S})$ and a function $f\!:\!X\rightarrow\!Y$, we say that $f$ is \emph{universally measurable} if $f^{-1}(S)\!\in\! \textnormal{UM}(X)$ for all $S\!\in\! \mathcal{S}$.
\end{defi}

The following facts will be useful.
\begin{thm}\label{properties_univ_meas}
Let $X$ be a Polish space. The following assertions hold:
\begin{enumerate}[\em(1)]
\item $\mbox{\boldmath$\Sigma$}^{1}_{1}(X)\subseteq \textnormal{UM}(X)$, and if $X$ is uncountable the inclusion is strict.
\item If $A\!\in\!\textnormal{UM}(Y)$ and $f\!:\!X\rightarrow Y$ is universally measurable, then $f^{-1}(A)\!\in\!\textnormal{UM}(X)$.
\end{enumerate}
\end{thm}

It is not possible, in \textnormal{ZFC} alone, to show that universally measurable sets extend any further up in the projective hierarchy.
\begin{thm}
Let $X$ be a Polish space. If $\textnormal{ZFC}$ is consistent the following assertions hold:
\begin{enumerate}[\em(1)]
\item $\textnormal{ZFC}\nvdash \mbox{\boldmath$\Delta$}^{1}_{2}(X)\subseteq \textnormal{UM}(X)$,
\item $\textnormal{ZFC}\nvdash \mbox{\boldmath$\Delta$}^{1}_{2}(X)\subsetneq \textnormal{UM}(X)$.
\end{enumerate}
\end{thm}

\begin{proof}
The result of the first assertion is due to Kurt G\"{o}del, see, e.g.,  Corollary 25.28 of \cite{Jech}. A proof of the second assertion can be found in, e.g., \cite{MS70}.
\end{proof}

Thus, it is not decidable in $\textnormal{ZFC}$ if every set in $\mbox{\boldmath$\Delta$}^{1}_{2}(X)$ is universally measurable. As we shall see in later sections, the winning sets of pL$\muprod$ games are, in general,  $\mbox{\boldmath$\Delta$}^{1}_{2}(X)$ sets. To deal with the associated measure theoretic issues, we will carry our main result in $\textnormal{ZFC}+\textnormal{MA}_{\aleph_{1}}$ set theory, where $\textnormal{MA}_{\aleph_{1}}$ is the so-called \emph{Martin's Axiom at $\aleph_{1}$}, the first uncountable cardinal. We often identify $\aleph_{1}$ with the least uncountable ordinal $\omega_{1}$. We refer to \cite{MS70,MioThesis} for a description of the axiom  $\textnormal{MA}_{\aleph_{1}}$ and some of its set-theoretic consequences. Here we list those that are relevant to our work.

\begin{thm}[\cite{MS70}]\label{consequences_martin}
Let $X$ be a Polish space, $\mu$ be a Borel probability measure on $X$ and $\{ A_{\alpha}\}_{\alpha<\omega_{1}}$ a collection of $\mu$-measurable subsets of $X$, where $\omega_{1}$ is the least uncountable ordinal. The following assertions hold in $\textnormal{ZFC}+ \textnormal{MA}_{\aleph_{1}}$:
\begin{enumerate}[\em(1)]
\item  $\mbox{\boldmath$\Delta$}^{1}_{2}(X)\subseteq \textnormal{UM}(X)$,
\item $\omega_{1}$-completeness:  $\bigcup_{\alpha<\omega_{1}}A_{\alpha}$ is a $\mu$-measurable set,
\item $\omega_{1}$-continuity: $\mu\big( \bigcup_{\alpha< \omega_{1}}A_{\alpha} \big) = \bigsqcup_{\alpha< \omega_{1}} \mu(A_{\alpha})$.
\end{enumerate}
\end{thm}
Note that, since singleton sets are always closed (hence measurable) in Polish spaces, it follows immediately from Theorem \ref{consequences_martin}  that Martin's Axiom at $\aleph_{1}$ implies $2^{\aleph_{0}}\neq \aleph_{1}$, i.e., the negation of \emph{Continuum Hypothesis} \cite{Jech}. As a matter of fact, Martin's Axiom was introduced by the authors as a possible set-theoretic alternative to the Continuum Hypothesis \cite{MS70}.

\section{The logic pL$\muprod$}
\label{syntax}
Given a set $\mathit{Var}$ of propositional variables ranged over by the letters $X$, $Y$ and $Z$, and a set of labels $L$ ranged 
over by the letters $a$, $b$ and $c$, the formulas of the logic pL${\muprod}$ are defined by the following 
grammar: 
\begin{center}
$F, G ::= X \   | \ F\wedge G\ | \ F \vee G\ | \  F\cdot G\ | \ F \odot G\ | \ \quadrato{a}F\ | \  \diam{a}F\ | \ \nu X.F \ | \ \mu X. F  $
\end{center}
which extends the syntax of the probabilistic modal $\mu$-calculus (pL$\mu$) with a new kind of conjunction ($\cdot$) and disjunction ($\odot$) operators called \emph{product}  and \emph{coproduct}  respectively.
As usual the operators $\nu X.F$ and $\mu X.F$ bind the variable $X$ in $F$. A formula is \emph{closed} if it has no \emph{free} variables. 

\begin{defi}[Subformula] We define the set $Sub(F)$ by case analysis on $F$ as follows: $Sub(   X  )\! \bydef\! \{  X   \}$, $Sub(   \quadrato{a}F  ) \! = \!  \{ \quadrato{a}F    \} \cup Sub(F)$, $Sub(  F_{1}\wedge F_{2}   )\! =\!  \{  F_{1}\wedge F_{2}   \} \cup Sub(F_{1} ) \cup Sub(F_{2})$, $Sub(  F_{1}\cdot F_{2}   ) \! = \!  \{  F_{1}\cdot F_{2}   \} \cup Sub(F_{1} ) \cup Sub(F_{2})$ and $Sub(  \nu X.F   )\! = \!  \{   \nu X.F \} \cup Sub(F)$. The cases for the connectives $\diam{a}$, $\vee$, $\odot$ and $\mu X$ are defined as for their duals. We say that $G$ is a \emph{subformula} of $F$ if $G\in Sub(F)$.
\end{defi}

Given a PLTS $\langle P, \{ \freccia{a} \}_{a\in L}\rangle$ we denote with $[0,1]^{P}$   the complete lattice of functions from $P$  to the real interval $[0, 1]$ with the pointwise order. A function $\rho\!:\! \mathit{Var}\!\rightarrow\! [0,1]^P$ is called a $[0,1]$-valued \emph{interpretation}, or just an interpretation, of the variables. Given a function $f\!:\! P\!\rightarrow\! [0,1]$ we denote with $\rho[f /X]$ the interpretation that assigns $f$ to the variable $X$, and $\rho (Y)$ to all other variables $Y$.

The denotational semantics $\sem{F}_{\rho}\!:\! P\! \rightarrow\! [0,1]$ of the pL${\muprod}$ formula $F$, under the interpretation $\rho$,  is defined by structural induction on $F$ as follows:
\begin{center}
\begin{tabular}{l l l}
$\sem{X}_{\rho}(p)$ & $=$ & $\rho(X)(p)$\\
$\sem{G \vee H}_{\rho}(p)$ & $=$  & $\sem{ G}_{\rho}(p) \sqcup \sem{H}_{\rho}(p)$\\
$\sem{G \wedge H}_{\rho}(p)$ & $=$  & $\sem{ G}_{\rho}(p) \sqcap \sem{H}_{\rho}(p)$\\
$\sem{G \odot H}_{\rho}(p)$ & $=$ & $\sem{ G}_{\rho}(p) \odot \sem{H}_{\rho}(p)$\\
$\sem{G \cdot H}_{\rho}(p)$ & $=$ & $\sem{ G}_{\rho}(p) \cdot \sem{H}_{\rho}(p)$\\
$\sem{\diam{a}G}_{\rho}(p) $ & $=$ & $  \displaystyle \bigsqcup_{p\freccia{a}d} \Big( \sum_{q\in \textnormal{supp}(d)} d(q) \cdot \sem{G}_{\rho}(q) \Big)$ \\ 
$\sem{\quadrato{a}G}_{\rho}(p) $ & $=$ & $  \displaystyle \bigsqcap_{p\freccia{a}d} \Big( \sum_{q\in \textnormal{supp}(d)} d(q) \cdot \sem{G}_{\rho}(q) \Big)$ \\ 
$\sem{\mu X. G}_{\rho}(p)$ & $=$ & $   \lfp \Big( \lambda f. ( \sem{G}_{\rho[ f/X]}) \Big) (p)$ \\
$\sem{\nu X. G}_{\rho}(p) $ & $=$ & $   \gfp \Big( \lambda f. ( \sem{G}_{\rho[ f/X]}) \Big) (p)$ 
\end{tabular}
\end{center}
where the symbols $\cdot$ and $\odot$ on the right hand side denote  the operations of product and coproduct on $[0,1]$, respectively.  It is easy to verify that the interpretation assigned to every pL$\muprod$ operator is monotone, thus the existence of the least and greatest fixed points is guaranteed by the Knaster--Tarski theorem. 

The interpretation of the connectives of the logic pL$\mu^{\odot}$ (and its fragment pL$\mu$) resembles the corresponding ones for L$\mu$. Both operations $\{\sqcup,\odot\}$, when restricted to the two element set $\{0,1\}$ act as ordinary boolean disjunction. Similarly, the operations $\{\sqcap,\cdot\}$ restricts to ordinary boolean conjunction. For what concerns the semantics associated with the modal operators, since in PLTS's transitions lead to probability distributions over states, rather than states, the most natural way to interpret the meaning of a formula $G$ at a probability distribution $d$ is to consider the \emph{expected probability} of the formula $G$ holding at a state $q$ randomly drawn in accordance with $d$, and this is formalized by the weighted sums in the definition above. 

\begin{rem}\label{remark_negation}
As it is common practice  when dealing with fixed point logics such as the modal $\mu$-calculus, we presented the syntax of pL$\mu^{\odot}$ in \emph{positive form}, i.e., without including a negation operator. This simplifies the presentation of the denotational semantics because all formulas in positive form are interpreted as monotone functions. A negation operator on (closed) pL$\mu^{\odot}$ formulas can be defined by induction on the structure of the formula, by exploiting the dualities between the connectives of the logic, in such a way that $\sem{\neg F}_{\rho}(p)= 1-\sem{F}_{\rho}(p)$, for all formulas $F$ and states $p$. We omit the routine details.
\end{rem}

As anticipated in the introduction, the main reason for extending the logic pL$\mu$ with the new pair of connectives of product and coproduct is to get a richer and more expressive logic capable of encoding the \emph{qualitative threshold modalities} defined as follows.
\begin{defi}\label{threshold_modalities_def}
Given a pL$\mu^{\odot}$ formula $F$, we define the macro formulas $\mathbb{P}_{>0}F$ and $\mathbb{P}_{=1}F$ as follows: $\mathbb{P}_{>0}F\bydef \mu X. (F \odot X)$ and $ \mathbb{P}_{=1}F\bydef \nu X. (F \cdot X)$, where $X$ is not free in $F$.
\end{defi}

The following lemma captures the denotational semantics of the qualitative threshold modalities.
\begin{lem}\label{qualitative_modalities_semantics}
Given a PLTS $\lts\!=\! \langle P, \{\freccia{a}\}_{a\in L}\rangle$, a $[0,1]$-interpretation of the variables $\rho\!\in\! \mathit{Var}\!\rightarrow \! [0,1]^{P}$ and a pL$\mu^{\odot}$ formula $F$, the following assertions hold:
\begin{center}
$\sem{\mathbb{P}_{>0}F}_{\rho}(p)\!=\! \left\{     \begin{array}{l  l}
 						1 & $if $\sem{F}_{\rho}(p) > 0\\
						0 & $otherwise$\\
						\end{array}      \right.
\ \ \ \ $and$ \ \ \ \ 
\sem{\mathbb{P}_{=1}F}_{\rho}(p)\!=\! \left\{     \begin{array}{l  l}
 						1 & $if $\sem{F}_{\rho}(p) =1 \\
						0 & $otherwise$\\
						\end{array}      \right.$\\
\end{center}
for every state $p\!\in\! P$.
\end{lem}

\begin{proof}
The map $x\mapsto \lambda \odot x$, for a fixed $\lambda\!\in\![0,1]$, has $1$ as unique fixed point when $\lambda\!>\!0$, and $0$ as the least fixed point when $\lambda\!=\!0$.  Similarly for the map $x\mapsto \lambda \cdot x$. The result then follows trivially.
\end{proof}

It is easy to verify that the  derived qualitative threshold modalities, if taken as primitives, are De Morgan duals. As an immediate consequence of Lemma \ref{qualitative_modalities_semantics} we have the following fact.
\begin{prop}\label{not_continuous_proposition}
The denotational interpretation of an open pL$\mu^{\odot}$ formula $F$ is, in general, not continuous  in the free variables, i.e., the denotation of a formula, seen as a function of type $\big(\mathit{Var}\!\rightarrow\! [0,1]^P\big)\!\rightarrow\![0,1]^{P}$, where $\big(\mathit{Var}\!\rightarrow\! [0,1]^P\big)$ is ordered pointwise, does not preserve countable $\sqsubseteq$-increasing chains.
\end{prop}

\begin{proof}
Consider the formula $\mathbb{P}_{=1}X\!=\!\nu Y. (X \cdot Y)$ having just one free variable $X$. We have that $\sem{\mathbb{P}_{=1}X}_{\rho}(p)\!=\!0$ if $\rho(X)(p)\!<\!1$ and  $\sem{\mathbb{P}_{=1}X}^{\lts}_{\rho}(p)\!=\!1$ if $\rho(X)(p)\!=\!1$.
\end{proof}

This contrasts with the fact that the denotations of pL$\mu$ formulas (i.e., formulas without occurrences of (co)products operators) when interpreted over finite PLTS's, are continuous in the free variables (see, e.g., Appendix C of \cite{MM07} for a proof of this fact, and \cite[\S 3.3.2.3]{MioThesis} for a discussion about this phenomenon). It then follows that the qualitative threshold modalities are not expressible in pL$\mu$. Thus, pL$\mu^{\odot}$ is a strictly more expressive logic than pL$\mu$, as previously claimed. 

The following theorem summarizes some expressivity results about the logic pL$\mu^{\odot}$. Since the focus of this paper is primarily  the study of a game semantics for pL$\mu^{\odot}$, we just refer to Theorem 7.2.16 and Proposition 7.2.5 in \cite{MioThesis} for detailed proofs.
\begin{thm}
The following propositions hold. s.
\begin{enumerate}[\em(1)]
\item the logic pL$\mu^{\odot}$ can encode the qualitative fragment of the logic \textnormal{PCTL} of \cite{BA1995},
\item there are (closed) pL$\mu^{\odot}$ formulas that can by satisfied only by infinite PLTS's. Thus pL$\mu^{\odot}$ does not satisfy the (expected adaptation of) the \emph{finite model property},
\item there are (closed) pL$\mu^{\odot}$ formulas satisfiable by some PLTS's but not satisfiable by any non-probabilistic PLTS (see Definition \ref{PLTS}),
\end{enumerate}
where a PLTS $\lts\!=\!\langle P, \{\freccia{a}\}_{a\in L}\rangle$ satisfies a (closed) formula $F$, if there is some $p\in P$ such that $\sem{F}(p)\!=\!1$.
\end{thm}

\section{Stochastic tree games}\label{tree_games_section}
In this (unavoidably long) section we introduce a new class of games which we call \emph{two-player stochastic tree games}, or just $2\frac{1}{2}$-player \emph{tree games}. Stochastic tree games generalizes standard two-player turn-based stochastic games (see, e.g., \cite{ChatPhD}, \cite{KLT99}) by allowing a new class of \emph{branching states} on which the execution of the game in split in independent concurrent subgames. Formally, stochastic tree games games are infinite duration games played by Player $1$, Player $2$ and a third probabilistic agent named \emph{Nature}, on a \emph{game arena} $\mathcal{A}\!=\!\langle (S,E), (S_{1},S_{2},S_{N}, B), \pi \rangle$, where $(S,E)$ is a directed graph with countable set of vertices $S$ and transition relation $E$,  $(S_{1},S_{2},S_{N},B)$ is a partition of $S$ and $\pi\! :\! S_{N}\!\rightarrow\! \mathcal{D}(S)$.
The states in $S_{1}$, $S_{2}$, $S_{N}$ and $S_{B}$ are called \emph{Player $1$} states, \emph{Player $2$} states, \emph{probabilistic} states and \emph{branching} states respectively. We denote with $E(s)$, for $s\!\in\! S$, the set $\{ s^{\prime}\ | \ (s,s^{\prime})\in E\}$. As a technical constraint, we require\footnote{The constraint $\supp(\pi({s}))\!=\! E(s)$, which might look more natural, is unnecessarily restrictive since one just one want to impose that Nature's choices always belong to the set of successor states. Our relaxed constraint will be technically convenient, see, e.g., Footnote \ref{nota_successors}.} that $\supp(\pi({s}))\!\subseteq\! E(s)$, for every $s\!\in\! S_{N}$.

\begin{defi}[Paths in $\arena$]
We denote with $\mathcal{P}^{\omega}$ and $\mathcal{P}^{<\omega}$ the sets of infinite and finite paths in $\arena$. Given a finite path $\vec{s}\!\in\!\mathcal{P}^{<\omega}$ we denote with $\first(\vec{s})$ and $\last(\vec{s})$ the first and last state of  $\vec{s}$, respectively. Given a finite path $\vec{s}$ and a (finite or infinite) path $\vec{t}$ we write $\vec{s}\lhd \vec{t}$ if $\vec{s}$ is a (not necessarily proper) prefix of  $\vec{t}$ and, provided that $\first(\vec{t})\!\in\!E\big(\last(\vec{s}) \big)$, we write $\vec{s}.\vec{t}$ for the concatenation of the two paths.  We denote with $\mathcal{P}^{t}$ the set of finite paths ending in a terminal state, i.e., the set of paths $\vec{s}$ such that $E(\last(\vec{s}))\!=\!\emptyset$. Similarly, we denote with $\mathcal{P}^{<\omega}_{1}$ and $\mathcal{P}^{<\omega}_{1}$  the sets of finite paths ending in a state in $S_{1}$ and $S_{2}$, respectively.
  We denote with $\mathcal{P}$ the set $\mathcal{P}^{\omega}\cup \mathcal{P}^{t}$ and we refer to this set as the set of \emph{completed paths} in $\arena$. Given a finite path $\vec{s}\!\in\! \mathcal{P}^{<\omega}$, we denote with $O_{\vec{s}}$ the set of all completed paths having $\vec{s}$ as prefix. We consider the standard topology on $\mathcal{P}$ where the countable basis for the open sets is given by the clopen sets $O_{\vec{s}}$, for $\vec{s}\!\in\! \mathcal{P}^{<\omega}$. This is a $0$-dimensional Polish space \cite{MioThesis}.
\end{defi}

\begin{defi}[Tree in $\arena$]
A  \emph{tree} in the arena $\arena$ is a collection $T\!=\!\{ \vec{s}_{i} \}_{i\in I}$ of finite paths $\vec{s}_{i}\!\in\! \mathcal{P}^{<\omega}$, such that
\begin{enumerate}[(1)]
\item $T$ is down-closed: if $\vec{s}\!\in\! T$ and $\vec{t}\lhd \vec{s}$  then $\vec{t}\!\in\!T$.
\item $T$ has a root: there exists exactly one finite path $\{s\}$ of length one in $T$. The state $s$, denoted by $\rootnode(T)$, is called the root of the tree $T$.
\end{enumerate}
We  consider the nodes $\vec{s}$ of $T$ as labeled by the $\last$ function.  
\end{defi}

\begin{defi}[Uniquely and fully branching nodes of a tree]
A node $\vec{s}$ in a tree $T$ is said to be \emph{uniquely branching} in $T$
if either $E(\last(\vec{s}))\! =\! \emptyset$ or $\vec{s}$ has a unique
child in $T$. Similarly, $\vec{s}$ is \emph{fully branching} in $T$
if, for every $s\! \in\! E(\last(\vec{s}))$, it holds that $\vec{s}.s\!\in\!T$.
\end{defi}

An outcome of the game in $\arena$, which we call a \emph{branching play}, is a possibly infinite tree $T$ in $\arena$ defined as follows:
\begin{defi}[Branching play in $\arena$]\label{branching_play} A \emph{branching play} in the arena $\arena$ is a tree $T$ in $\arena$ such that, for every node $\vec{s}\!\in\! T$ the following conditions holds:
\begin{enumerate}[(1)]
\item If $\last(\vec{s})\!\in\! S_{1}\cup S_{2}\cup S_{N}$  then $\vec{s}$ branches uniquely in $T$.
\item If $\last(\vec{s})\!\in\! B$ then $\vec{s}$ branches fully in $T$.
\end{enumerate}
We denote with $\bp$ the set of branching plays $T$ in the arena $\arena$. 
\end{defi}
A branching play $T$ represents a possible execution of the game from the state $s$ labeling the root of $T$. The nodes of $T$ with more than one child are all labeled with a state $s\!\in\! B$ and are the branching points of the game. Their children represent the many independent instances of play generated at the branching point.

\begin{defi}[Topology on $\bp$] \label{topology_bp}Given a finite tree $F$ in $\arena$, we denote with $O_{F}\subseteq\bp$ the set of all branching plays $T$ such that $F\subseteq T$. We fix the topology on $\bp$, where the basis for the open sets is given by the clopen sets $O_{F}$, for every branching-play prefix $F$. It is routine to show that this is a $0$-dimensional Polish space \cite{MioThesis}.
\end{defi}
As usual when working with \emph{stochastic} games, it is useful  to look at the possible outcomes of a play up-to the behavior of Nature. In the context of standard two-player stochastic games this amounts to considering Markov chains. In our setting the following definition of Markov branching play is natural:

\begin{defi}[Markov branching play in $\arena$]\label{markov_branching_play} A \emph{Markov branching play} in $\arena$ is a tree $M$ in $\arena$ such that  for every node $\vec{s} \!\in\! M$, the following conditions hold:
\begin{enumerate}[(1)]
\item If $\last(\vec{s})\!\in\! S_{1}\cup S_{2}$  then $\vec{s}$ branches uniquely in $T$.
\item If $\last(\vec{s})\!\in\! S_{N}\cup B$ then $\vec{s}$ branches fully in $T$.
\end{enumerate}
\end{defi}
A Markov branching play is similar to a branching play except that probabilistic choices of Nature have not been resolved. 

\begin{defi}[Probability measure $\mathbb{P}_{M}$]
\label{measure_definition}
Every Markov branching play $M$ determines a probability assignment $\mathbb{P}_{M}(O_{F})$ to every basic clopen set $O_{F}\!\subseteq\!\bp$, for $F$ a finite tree in $\arena$ (we can assume that every node $\vec{s}\!\in\!F$ such that $\last(\vec{s})\!\in\!S_{N}$ has a unique child in $F$, since otherwise $O_{F}\!=\!\emptyset$),  defined as follows:

\begin{center}
 $\mathbb{P}_{M}(O_{F}) \bydef	 \left\{     \begin{array}{l  l}  \displaystyle \prod \{Ê\pi (s)(s^{\prime}) \ | \  {\vec{s}.s.s^{\prime}}\! \in\! F  \wedge s\!\in\! S_{N}\}\ \  & $if $ F\subseteq M\\
 											 0 & $otherwise$
						 \end{array}      \right.$

\end{center} 
Such an assignment extends, by Carath\'{e}odory's Extension Theorem \cite{Tao_measuretheory,MioThesis}, to
 a unique Borel probability measure on $\bp$, whence to a complete probability measure, also denoted by $\mathbb{P}_{M}$.
\end{defi}

It is the definition above that  implements  the \emph{probabilistic independence} of the sub-branching plays  rooted at some branching node.

\begin{defi}[Two-player stochastic tree game]\label{tree_game_def}
A \emph{two-player stochastic tree game} (or a $2\frac{1}{2}$-player tree game) is given by a pair $\langle \arena, \Phi\rangle$, where $\arena$ is a stochastic tree game arena as described above, and $\Phi\!\subseteq\!\bp$, which is the \emph{objective} or \emph{winning set} for Player $1$, is a universally measurable set of branching plays in $\arena$.
\end{defi}

\begin{defi}[Expected value of a Markov branching play]
\label{value_markov_play}
Let $\langle \arena, \Phi\rangle$ be a $2\frac{1}{2}$-player tree game, and $M$ a Markov branching play in $\arena$. We define the \emph{expected value}  of $M$ as follows:  $\expected(M)\!=\! \mathbb{P}_{M}(\Phi)$. The value $\expected(M)$ should be understood as the probability for Player $1$ to win the probabilistic play represented by $M$.
\end{defi}
As usual in game theory, players' moves are determined by strategies. 

\begin{defi}[Deterministic strategies]
An (\emph{unbounded memory deterministic}) strategy $\sigma_{1}$ for Player $1$ in $\arena$ is defined as a function $\sigma_{1}\!:\!\mathcal{P}_{1}^{<\omega}\!\rightarrow\! S\cup \{ \bullet  \}$ such that $\sigma_{1}(\vec{s})\!\in\! E(\last(s))$ if $E(\last(\vec{s}))\!\not =\! \emptyset$ and $\sigma_{1}(\vec{s})\!=\! \bullet$ otherwise. Similarly a \emph{strategy} $\sigma_{2}$ for Player $2$ is defined as a function $\sigma_{2}\!:\!\mathcal{P}_{2}^{<\omega}\!\rightarrow\! S\cup \{ \bullet  \}$. A pair $\langle \sigma_{1},\sigma_{2}\rangle$ of strategies, one for each player, is called a  \emph{strategy profile} and determines the behaviors of both players.
\end{defi}
Note that the above definition of strategy captures the intended behavior of the game: both players, when acting on a given instance of the game, know all the history of the actions happened on that subgame, but have no knowledge of the evolution of the other \emph{independent} parallel subgames.

\begin{defi}[$M^{s_{0}}_{\sigma_{1},\sigma_{2}}$] Given an initial state $s_{0}\!\in\! S$ and a strategy profile $\langle\sigma_{1},\sigma_{2}\rangle$ a unique Markov branching play  $M^{s_{0}}_{\sigma_{1},\sigma_{2}}$ is determined:
\begin{enumerate}[(1)]
\item the root of $M$ is labeled with $s_{0}$,
\item for every $\vec{s}\!\in\! M^{s_{0}}_{\sigma_{1},\sigma_{2}}$, if $\last(\vec{s})\!=\!s$ with $s\!\in\! S_{1}$ not a terminal state, then the unique child of  $\vec{s}$ in $M^{s_{0}}_{\sigma_{1},\sigma_{2}}$ is $\vec{s}.\big(\sigma_{1}(\vec{s})\big)$,
\item for every $\vec{s}\!\in\! M^{s_{0}}_{\sigma_{1},\sigma_{2}}$, if $\last(\vec{s})\!=\!s$ with $s\!\in\! S_{2}$ not a terminal state, then the unique child of  $\vec{s}$ in $M^{s_{0}}_{\sigma_{1},\sigma_{2}}$ is $\vec{s}.\big(\sigma_{2}(\vec{s})\big)$.
\end{enumerate}
This specifies uniquely $M^{s_{0}}_{\sigma_{1},\sigma_{2}}$  because Markov branching plays branch fully on probabilistic and branching states. 
\end{defi}

\begin{defi}[Upper and lower values of a $2\frac{1}{2}$-player tree game]\label{deterministic_value}
Let $\game\!=\!\langle \arena, \Phi\rangle$ be a $2\frac{1}{2}$-player tree game. We define the lower and upper values of $\game$ on the state $s$, denoted by $\val^{s}_{\downarrow}(\game)$ and $\val^{s}_{\uparrow}(\game)$ respectively, as follows:
\begin{center}
$\val^{s}_{\downarrow}(\game)= \bigsqcup_{\sigma_{1}}\bigsqcap_{\sigma_{2}}\!\! \expected (M^{s}_{\sigma_{1},\sigma_{2}}) \ \ \ \ \ \ \ \ \val^{s}_{\uparrow}(\game)= \bigsqcap_{\sigma_{2}}\bigsqcup_{\sigma_{1}}\!\! \expected(M^{s}_{\sigma_{1},\sigma_{2}})$
\end{center}
\end{defi}
$\val^{s}_{\downarrow}(\game)$ represents the limit probability of Player $1$ winning, when the game begins at $s$, by choosing his strategy $\sigma_{1}$ first and then letting Player $2$ pick an appropriate counter strategy $\sigma_{2}$. Similarly $\val^{s}_{\uparrow}(\game)$ represents the limit probability of Player $1$ winning, when the game begins at $s$, by first letting Player $2$ choose a strategy $\sigma_{2}$ and then picking an appropriate counter strategy $\sigma_{1}$. Clearly, for every $s$, the following inequality holds: $\val^{s}_{\downarrow}(\game) \leq \val^{s}_{\uparrow}(\game)$. In the special case (not true in general) that this inequality is an equality, we say that the game $\game$  is \emph{determined} at $s$.

\begin{defi}[$\varepsilon$-optimal strategies]
\label{epsilon_optimal}
Let $\game\!=\!\langle \mathcal{A}, \Phi\rangle$ be a $2\frac{1}{2}$-player tree game. We say that a strategy $\sigma_{1}$  for Player $1$ in $\game$ is $\varepsilon$-optimal, for $\varepsilon\!\geq\!0$, if for every state $s$, the following inequality holds: 
\begin{center}
$\bigsqcap_{\sigma_{2}} \!\! \expected (M^{s}_{\sigma_{1},\sigma_{2}}) \geq \val_{\downarrow}^{s}(\game) - \varepsilon$.
\end{center}
Similarly we say that a deterministic strategy $\sigma_{2}$ for Player $2$ in $\game$ is  $\varepsilon$-optimal, for $\varepsilon\!\geq\!0$, if  for every state $s$, the following inequality holds: 
\begin{center}
$\bigsqcup_{\sigma_{1}} \!\! \expected (M^{s}_{\sigma_{1},\sigma_{2}}) \leq \val_{\uparrow}^{s}(\game) + \varepsilon$.
\end{center}
Clearly $\varepsilon$-optimal (mixed and deterministic) strategies for Player $1$ and Player $2$ always exist for every $\varepsilon\! >\!0$, but not necessarily so for $\varepsilon\! =\!0$.  
\end{defi}

\begin{rem}\label{remark_tree_games_without_branching_states}
Observe that a $2\frac{1}{2}$-player tree game without branching states is just an ordinary two-player  turn-based stochastic game (see, e.g., \cite{KLT99}, \cite{MioThesis}): the set of branching plays is homeomorphic to the set of completed paths, the notion of Markov branching play collapses to the standard notion of Markov chain, and strategies are maps from finite paths to successor states. Thus, as previously claimed, two-player stochastic tree games constitute a generalization of ordinary two-player stochastic games. 
\end{rem}

Notwithstanding our primary interest in  an appropriate game semantics for pL$\mu^{\odot}$, we highlighting here that the simple form of  \emph{partial information} implemented in  tree games (players when acting on a given subgame are not aware of what happens in the other independent subgames) is surprisingly powerful. For instance the class of Blackwell games \cite{blackwell69,Martin98} can be encoded as two-players tree games (Theorem 4.2.18 in \cite{MioThesis}). Moreover, interesting open problems in the field, such as that of \emph{qualitative determinacy} of stochastic games (see, e.g., \cite{vaclav2011a}), can be formulated as appropriate determinacy problems for $2$-player tree games (see, \cite[\S 4.4]{MioThesis}). We refer to the author's PhD thesis \cite{MioThesis} for an extensive introduction and analysis of tree games.

\subsection{Two-player stochastic meta-parity games}\label{section_meta_parity}
In this subsection we identify a class of two-player stochastic tree games, called $2\frac{1}{2}$-player meta-parity games, which will be used to give an appropriate game semantics to the logic pL$\mu^{\odot}$.

A  $2\frac{1}{2}$-player meta-parity game is  a $2\frac{1}{2}$-player tree game $\game\!=\!\langle \arena,\Phi\rangle$, which we refer to as the \emph{outer} game, in which every branching play $T$ is itself interpreted as a (ordinary) two-player parity game $G_{T}$, which we refer to as the \emph{inner game associated with $T$}, and whose objective $\Phi$  is defined as the set of branching plays $T$ for which Player $1$ has a winning strategy in $G_{T}$.

We start formalizing this notion with the following definitions.
\begin{defi}\label{parity_assignment_def}
A \emph{parity assignment} $\pr$ for a two-player stochastic tree game arena $\arena\!=\!\langle (S,E), ( S_{1},S_{2},S_{N}, B ), \pi \rangle$ is a function $\pr\!:\! S\!\rightarrow\! \mathbb{N}$ whose image is finite. In other words $\pr$ assigns to each state $s\!\in\! S$ a natural number, also referred to as a \emph{priority},  taken from a finite pool of options $\{n_{0}, . . . , n_{k}\} \!=\! \pr(S)$. We denote with $\max(\pr)$, $\min(\pr)$ and $|\pr|$ the natural numbers $\max\{n_{0},...,n_{k}\}$, $\min\{n_{0}, . . . , n_{k}\}$ and $|\{n_{0}, . . . , n_{k}\}|$ respectively.
\end{defi}

The function $\pr$ induces a set of completed paths, denoted by $\mathcal{W}_{\pr}$, specified as follows.

\begin{defi}\label{WPR}
Let $\arena\!=\!\langle (S,E),(S_{1},S_{2},S_{N},B),\pi\rangle$ be a $2\frac{1}{2}$-player tree game arena and $\pr$  a parity assignment for it.
A completed path $\vec{s}$ belong to the \emph{parity set induced by $\pr$}, denoted by  $\mathcal{W}_{\pr}\!\subseteq\!\mathcal{P}$, if either:
\begin{enumerate}[(1)]
\item $\vec{s}$ is a finite terminated path, i.e., $\vec{s}\!\in\!\mathcal{P}^{t}$, and the priority assigned to the last state of $\vec{s}$ is odd, i.e., $\pr\big(\last(\vec{s})\big)\! \equiv\! 1 \!\pmod{2}$, or
\item $\vec{s}$ is infinite, i.e., $\vec{s}\!\in\!\mathcal{P}^{\omega}$ with $\vec{s}\!=\! \{s_{i}\}_{i\in\mathbb{N}}$, and the greatest priority assigned to infinitely many states $s_{i}$ in $\vec{s}$ is even, i.e., $\big(\displaystyle \limsup_{i\in \mathbb{N}} \pr(s_{i})\big)  \equiv\! 0 \!\pmod{2}$.
\end{enumerate}
It is well known (see, e.g., \cite{ChatPhD}) that $\mathcal{W}_{\pr}$ is a  $\mbox{\boldmath$\Delta$}^{0}_{3}$ set, hence a Borel set.
\end{defi}

\begin{defi}
A \emph{player assignment} $\pl$ for a two-player stochastic tree game arena $\arena\!=\!\langle (S,E), ( S_{1},S_{2},S_{N}, B ), \pi \rangle$ is a function $\pl\!:\! B\!\rightarrow\! \{1,2\}$. We often include the information provided by $\pl$ directly in the signature of $\arena$ by considering the partition $( S_{1},S_{2},S_{N}, B_{1},B_{2})$, where $B_{i}\!=\!\pl^{-1}(\{i\})$, for $i\!\in\!\{1,2\}$. 
\end{defi}

The function $\pl$ assigns a \emph{player identifier} to each state $s\!\in\! B$. This allows to consider each branching play $T$ in $\arena$ as  a parity game $G_{T}$  (induced by $\pr$ and $\pl$) played  by Player $1$ and Player $2$ on the tree $T$, where Player $1$ and Player $2$ controls the vertices $\vec{s}$ of $T$  such that $\pl(\last(\vec{s}))\!=\!1$ and  $\pl(\last(\vec{s}))\!=\!2$ respectively. All other vertices are either leaves, in which case the game $G_{T}$ ends, or have a unique child, towards which the game automatically progresses. The result of the game is a branch in $T$ or, equivalently, a completed path in $\arena$. Adopting standard terminology (see, e.g., \cite{Martin75}, \cite{Stirling96}), we say that Player $1$ (respectively Player $2$) wins the parity game $G_{T}$ if they have a winning strategy guaranteeing the outcome of the game to be in the set $\mathcal{W}_{\pr}$ (respectively, in $\mathcal{P}\setminus\mathcal{W}_{\pr}$). It is well known (see, e.g., \cite{Martin75}), that one of the two player has a winning strategy in $G_{T}$.

We are now ready to formally introduce the notion of $2\frac{1}{2}$-player meta-parity game.
\begin{defi}\label{def_metaparity}
Given a $2\frac{1}{2}$-player tree game arena, a priority assignment $\pr$ and a player assignment $\pl$ for it, the associated \emph{two-player stochastic meta-parity game} is defined as the $2\frac{1}{2}$-player tree game $\game\!=\!\langle \arena,\Phi_{\pr,\pl}\rangle$, where $\Phi_{\pr,\pl}\!\subseteq\!\bp$ is defined as follows:
\begin{center}
$\Phi_{\pr,\pl}= \{ T \ | \ T\in\bp$ and Player $1$ has a winning strategy in $G_{T}\}$.
\end{center} 
Note that, by previous observations, the set $\bp\setminus \Phi_{\pr,\pl}$ can be specified as follows:
\begin{center}
$\bp\setminus \Phi_{\pr,\pl}= \{ T \ | \ T\in\bp$ and Player $2$ has a winning strategy in $G_{T}\}$.
\end{center} 
We often just write $\Phi$ if the priority and player assignments  are clear from the context.
\end{defi}

\begin{rem}\label{remark_meta_games_without_branching_states}
Note that a $2\frac{1}{2}$-player meta-parity game $\game\!=\!\langle \arena,\Phi_{\pr,\pl}\rangle$ without branching states is just an ordinary $2\frac{1}{2}$-player game. The game $G_{T}$ associated with a branching play $T$ (which is just a completed path $\vec{s}$, see Remark \ref{remark_tree_games_without_branching_states}) is trivial, and belongs to   $\Phi_{\pr,\pl}$ if and only if $\vec{s}\!\in\!\mathcal{W}_{\pr}$. Thus $2\frac{1}{2}$-player meta-parity games generalize ordinary $2\frac{1}{2}$-player parity games in an obvious way.
\end{rem}

Note that Definition \ref{def_metaparity} is meaningful, in accordance with Definition \ref{tree_game_def}, only if the set $\Phi_{\pr,\pl}$ is a universally measurable set of branching plays. As we now discuss, this turns out to be a delicate point.

\begin{thm}\label{delta_complexity}
Given a $2\frac{1}{2}$-player tree game arena $\arena$, a priority assignment $\pr$ and a player assignment $\pl$ for it, the associated set $\Phi_{\pr,\pl}$ of branching plays is a $\mbox{\boldmath$\Delta$}^{1}_{2}$ set.
\end{thm}

\begin{proof}
Let us denote with $\mathcal{P}^{<\omega}_{B_i}$ the set of finite paths $\vec{s}\!\in\!\mathcal{P}^{<\omega}$ in $\arena$ such that $\last(\vec{s})\!\in\! B$ and $\pl(\last(\vec{s}))\!=\! i$, for $i\!\in\!\{1,2\}$. Let us consider the set $\Sigma_{i}$ of functions $\mathcal{P}_{B_i}^{<\omega}\!\rightarrow\! \mathcal{P}^{<\omega}\!\cup\!\{ \bullet \}$. This set contains all the strategies available to Player $i$ in every game $\game_{T}$, for $T\!\in\! \bp$, seen as functions $f\!\in\! \Sigma_{i}$ restricted to $T$.  We endow $\Sigma_{i}$ with the Baire space-like topology, where for every pair $(x,y)$, with $x\!\in\!\mathcal{P}^{<\omega}_{B_i}$ and $y\!\in\!\mathcal{P}^{<\omega}\!\cup\! \{\bullet\}$, the set $O_{x,y}$ of all functions $f\!\in\!\Sigma_{i}$ such that $f(x)\!=\!y$ is a basic open set. This is a $0$-dimensional Polish space. 

Let us now consider the subset of  $\bp\times \Sigma_{\uno} \times \Sigma_{\due}$, denoted by $\mathcal{T}$, consisting of all triples $(T,\sigma_{\uno},\sigma_{\due})$ such that the strategies $\sigma_{\uno}$ and $\sigma_{\due}$ are valid strategies in $\game_{T}$, i.e., functions from finite sequences of vertices in $T$ to vertices in $T$. It is easy to see that $\mathcal{T}$ is a closed subset of $\bp\times \Sigma_{\uno} \times \Sigma_{\due}$ (endowed with the product topology). Indeed, the set of triples which do not belong to $\mathcal{T}$ is open, because one can tell if one of the two strategies $\sigma_{\uno}$ and $\sigma_{\due}$ is not valid in the inner game $\game_{T}$, i.e., it makes choices which are not in $T$, just by looking at finite information about $T$, $\sigma_{\uno}$ and $\sigma_{\due}$. Hence $\mathcal{T}$ is a Polish space, as it is a closed subset of the Polish space $\bp\times \Sigma_{\uno} \times \Sigma_{\due}$.

Let us denote with $out\!:\! \mathcal{T} \rightarrow \mathcal{P}$ the function which maps a triple $(T,\sigma_{\uno},\sigma_{\due})$ to the induced play (i.e., a completed path in $\arena$) in $\game_{T}$. The function $out$ is clearly continuous.

Let us now consider the set $A\!\subseteq\! \mathcal{T}$ defined as the set of triples $(T,\sigma_{\uno}, \sigma_{\due})$  such that the game $G_{T}$ is won by Player $\due$ when the two players follow the (valid for $\game_{T}$) strategies $\sigma_{\uno}$ and $\sigma_{\due}$ respectively, i.e., the set formally defined as $A\!=\! out^{-1}(\mathcal{P}\setminus \mathcal{W}_{\pr})$.
Since $out$ is continuous and $\mathcal{W}_{\pr}$ is a Borel set, it follows that $A$ is a Borel set.

Let us now define the set $B\!\subseteq\! \bp\!\times\! \Sigma_{\uno}$ as $B\!=\!\{ (T,\sigma_{\uno}) \ | \ \exists \sigma_{\due}\!\in\! \Sigma_{\due}. (T,\sigma_{\uno},\sigma_{\due})\in A\}$. The set $B$ is the set of all pairs $(T,\sigma_{\uno})$,  such that Player $\due$ has a strategy $\sigma_{\due}$ in the game $\game_{T}$ winning against $\sigma_{\uno}$, i.e., such that the strategy profile $(\sigma_{\uno},\sigma_{\due})$ induces in $\game_{T}$ a completed path in $\mathcal{P}\setminus \mathcal{W}_{\pr}$. The set $B$ is a $\mbox{\boldmath$\Sigma$}^{1}_{1}$ set by construction.  Observe that its complement $\overline{B}$ is the set of all pairs $(T,\sigma_{\uno})$ such that Player $\due$ does not have a strategy $\sigma_{\due}$ for the game $\game_{T}$  winning against the strategy $\sigma_{\uno}$. Equivalently, by determinacy of $2$-player parity games \cite{Martin75}, the strategy  $\sigma_{\uno}$ is  a  winning strategy for Player $\uno$ in the game $\game_{T}$.  The set $\overline{B}$ is a $\mbox{\boldmath$\Pi$}^{1}_{1}$ set by construction. We can now define the set $\Phi_{\pr,\pl}\!\subseteq\!\bp$ of all branching plays $T$ where Player $\uno$ has a winning strategy in $\game_{T}$ as $\Phi_{\pr,\pl} = \{ T \ | \ \exists \sigma_{\uno}\!\in\! \Sigma_{\uno}. (T,\sigma_{\uno})\!\in\! \overline{B}\}$. It then follows that, by construction, $\Phi_{\pr,\pl}$ is a $\mbox{\boldmath$\Sigma$}^{1}_{2}$ set.

The desired result then follows  by observing that the complement set $\overline{\Phi_{\pr,\pl}}$ is also a $\mbox{\boldmath$\Sigma$}^{1}_{2}$ set. This is because $\overline{\Phi_{\pr,\pl}}$ is the winning set associated with the triplet ($\arena, \overline{\pr},\overline{\pl})$, where $\overline{\pr}$ is the dual parity assignment defined as $\overline{\pr}(s)=\pr(s)+1$, and $\overline{\pr}$ is the dual player assignment inverting the role of the two players, specified as $\overline{\pl}(b)\!=\!0$ if and only if $\pl(b)\!=\!1$.
\end{proof}

The following theorem asserts that the result of Theorem \ref{delta_complexity} is strict. Thus the technologies employed in this paper for dealing with the complexity of the winnings sets in $2\frac{1}{2}$-player meta-parity games are not trivially avoidable. Since the technicalities required for proving the result would not be particularly useful for the main theorem of this paper, we just provide a reference to a detailed proof.
\begin{thm}
There exists a $2\frac{1}{2}$-player meta-parity game $\game\!=\!\langle \arena,\Phi_{\pr,\pl}\rangle$, having a finite arena $\arena$, such that the winning set $\Phi_{\pr,\pl}$ is not analytic nor co-analytic, i.e., $\Phi_{\pr,\pl}\!\not\in\! \mbox{\boldmath$\Sigma$}^{1}_{1}\cup \mbox{\boldmath$\Pi$}^{1}_{1}$.
\end{thm}

\begin{proof}
One can construct an explicit example of finite $2\frac{1}{2}$-player meta-parity game and show that both a $\mbox{\boldmath$\Sigma$}^{1}_{1}$-complete set  and a $\mbox{\boldmath$\Pi$}^{1}_{1}$-complete set are Wadge-reducible to the winning set $\Phi_{\pr,\pl}$. We refer to Theorem 6.4.3 in \cite{MioThesis} for a detailed proof. 
\end{proof}

It follows that, absent any further evidence, the winning set of a  $2\frac{1}{2}$-player meta-parity game might be not provably universally measurable in $\textnormal{ZFC}$ set theory. However, as stated in Theorem \ref{consequences_martin}, it is provably universally measurable in $\textnormal{ZFC}+\textnormal{MA}_{\aleph_{1}}$ set theory.  This is one of the uses we make of Martin's axiom at $\aleph_{1}$ but, as we shall note later  in Section \ref{remarks_martin}, not the only one.

The following theorem, which will be useful later, exposes an important property of $2\frac{1}{2}$-player meta-parity games and sheds some light on the relationship between this class of games and the logic pL$\mu^{\odot}$. It constitutes the expected generalization of the corresponding property of ordinary $2\frac{1}{2}$-player parity games (see, e.g., Proposition 4.19 in \cite{MIO10}). 
\begin{thm}\label{fixed_point_proposition_1}
Let $\game\!=\!\langle \arena, \Phi_{\pr,\pl}\rangle$ be a two-player stochastic  parity game with arena $\arena\!=\! \langle (S,E),(S_{1},S_{2},S_{N}, B_{1},B_{2}),\pi\rangle$, where $B_{i}\!=\!\pl^{-1}(\{i\})$. The functions $\val_{\downarrow}(\game)$ and  $\val_{\uparrow}(\game)$, of type $S\!\rightarrow\![0,1]$, are fixed points of the  functional $\mathcal{F}\!:\![0,1]^{S}\!\rightarrow\![0,1]^{S}$ defined as follows:
\begin{center}
$\displaystyle \mathcal{F}(f)(s)= \left\{      \begin{array}{l  l}      \pr(s) \pmod 2  & $if $E(s)\!=\!\emptyset$, i.e, if $s$ is a terminal state$ \\ 
							               	\displaystyle 	\bigsqcup_{t\in E(s)}f(t) & $if$ \  s\!\in\! S_{1}\\
									\displaystyle 	\bigsqcap_{t\in E(s)}f(t) & $if$ \  s\!\in\! S_{2}\\
									\displaystyle 	\sum_{t\in E(s)} \pi(s)(t) \cdot f(t) & $if$ \  s\!\in\! S_{N}\\
									\displaystyle 	\coprod_{t\in E(s)}  f(t) & $if$ \  s\!\in\! B_{1}\\
									\displaystyle 	\prod_{t\in E(s)}  f(t) & $if$ \  s\!\in\! B_{2}\\
                      	      		     \end{array}      \right.$ 
\end{center}
\end{thm}

\begin{proof}
The proof is carried out following the same methodology of, e.g., Proposition 4.19 in \cite{MIO10}. The interesting cases are associated with the analysis of the branching nodes $s\!\in\!B_{1}$ and $s\!\in\! B_{2}$. In what follows we just show that, for every $s\!\in\! B_{1}$, the following equality holds:
\begin{center}
 $ \ \ \ \val_{\downarrow}(\game)(s)\!\geq\! \coprod_{t\in E(s)}\val_{\downarrow}(\game)(t)$.
 \end{center} 
The reverse inequality  $\val_{\downarrow}(\game)(s)\!\leq\!\coprod_{t\in E(s)}\val_{\downarrow}(\game)(t)$, as well as all other cases  for $s\!\in\! \{S_{1},S_{2},S_{N}B_{2}\}$ can be proved in a similar way. We refer to Theorem 5.2.10 of \cite{MioThesis} for a detailed proof covering all cases.

By  Definition \ref{deterministic_value} of $\val_{\downarrow}(\game)$, we need to prove that the equality  $\bigsqcup_{\sigma_{1}} \bigsqcap_{\sigma_{2}} \expected(M^{s}_{\sigma_{1},\sigma_{2}})   \!\geq\!\coprod_{t\in E(s)}\big( \bigsqcup_{\tau_{1}} \bigsqcap_{\tau_{2}} \expected(M^{t}_{\tau_{1},\tau_{2}}) \big) $ holds. Let $E(s)\!=\!\{t_{i}\}_{i\in I}$. At the state $s$ the game is split in $I$-many subplays continuing their execution from the states $t_{i}$. Let $\tau^{i}_{1}$ be a $\varepsilon_{i}$-optimal strategy (see Definition \ref{epsilon_optimal}) for Player $1$,  with $\varepsilon_{i}\!>\!0$.

Define the strategy $\sigma_{1}$ for Player $1$, when the game starts at $s$ and the $I$-many subplays are generated, to behave in the subplay continuing its execution from $t_{i}$ as the strategy  $\tau^{i}_{1}$. Given any strategy $\sigma_{2}$ for Player $2$, the Markov branching play $M^{s}_{\sigma_{1},\sigma_{2}}$ can be depicted\footnote{The edge connecting $s$ with $t_{i}$ has not been dashed  to highlight that $s$ is a branching node.}  as follows:
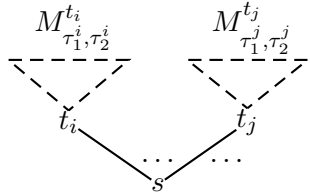
\begin{center}
\pstree[ treemode=U,levelsep=5ex]{\Tr{$s$}}{
	\pstree[levelsep=5ex,  treemode=U,linestyle=dashed]{\Tr{$t_{i}$}     \trput{$\dots$} }{
		\pstree[linestyle=none,arrows=-,levelsep=3ex]{\Tfan[fansize=10ex]}{\TR{ $M^{t_{i}}_{\tau^{i}_{1},\tau^{i}_{2}} $}}
	}
		\pstree[levelsep=5ex,  treemode=U,linestyle=dashed]{\Tr{$t_{j}$}    \trput{$\dots$}}{
		\pstree[linestyle=none,arrows=-,levelsep=3ex]{\Tfan[fansize=10ex]}{\TR{ $M^{t_{j}}_{\tau^{j}_{1},\tau^{j}_{2}} $}}
	}
}
\end{center}
where, for each $i\!\in\! I$, the strategy $\tau^{i}_{2}$ is specified as $\tau^{i}_{2}(\vec{t}_{i})\!=\! \sigma_{2}(s.\vec{t}_{i})$ for all finite paths $\vec{t}_{i}$ with $\first(\vec{t}_{i})\!=\!t_{i}$. 
Let us denote, to improve readability, with $M$ and $M_{i}$, for $i\!\in\!I$, the Markov branching plays $M^{s}_{\sigma_{1},\sigma_{2}}$ and $M^{t_{i}}_{\tau^{i}_{1},\tau^{i}_{2}} $, respectively. We are now going to show that the following equality holds:
\begin{equation}\label{eq_aux_1_fixedpoint_theorem}
\mathbb{P}_{M}(\Phi_{\pr,\pl})\!=\! \displaystyle \coprod_{i\in I} \mathbb{P}_{M_{i}}(\Phi_{\pr,\pl})
\end{equation}
holds. This will conclude the proof. Indeed note that, by construction of $\sigma^{i}_{1}$, the inequality $ \mathbb{P}_{M_{i}}(\Phi_{\pr,\pl})\!\geq\! \val_{\downarrow}(\game)(t_{i})-\varepsilon_{i}$ holds. Thus, the strategy $\sigma_{1}$ guarantees, by appropriate choices of values $\varepsilon_{i}$, for $i\!\in\! I$, a value closed to $\displaystyle \coprod_{i\in I}\val_{\downarrow}(\game)(t_{i})$ as desired.

For what concerns $\mathbb{P}_{M}$, we can restrict attention to the set of branching plays $\bp_{\!s}$, where $\bp_{\!s}$ denotes the set of branching plays in $\arena$ rooted at $s$, since the set of all other branching plays in $\game$ gets assigned probability $0$ by Definition \ref{measure_definition} of $\mathbb{P}_{M}$. Similarly, when considering $\mathbb{P}_{M_{i}}$, we can restrict to the set $\bp_{\!i}$ of branching plays rooted at $t_{i}$. We can depict the branching plays in $\bp_{\!s}$ and the branching plays in $\bp_{\!i}$ as follows:
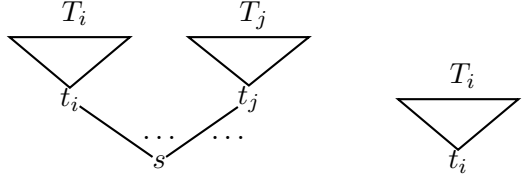
\begin{center}
\pstree[ treemode=U,levelsep=5ex]{\Tr{$s$}}{
	\pstree[levelsep=5ex,  treemode=U]{\Tr{$t_{i}$}     \trput{$\dots$} }{
		\pstree[linestyle=none,arrows=-,levelsep=2ex]{\Tfan[fansize=10ex]}{\TR{ $T_{i} $}}
	}
		\pstree[levelsep=5ex,  treemode=U]{\Tr{$t_{j}$}    \trput{$\dots$}}{
		\pstree[linestyle=none,arrows=-,levelsep=2ex]{\Tfan[fansize=10ex]}{\TR{ $T_{j} $}}
	}
}$ \ \ \ \ \ \ \ \ $
	\pstree[levelsep=5ex,  treemode=U]{\Tr{$t_{i}$}}{
		\pstree[linestyle=none,arrows=-,levelsep=2ex]{\Tfan[fansize=10ex]}{\TR{ $T_{i} $}}
	}
\end{center} 
where we use $T_{i}$ to range over the set of branching plays in $\game$ rooted at $t_{i}$. We denote with $s[T_{i}]_{i\in I}$ the branching play on the left. Let $\prod_{i}\bp_{\!i}$ be endowed with the product topology. Define $m\!:\! \prod_{i}\bp_{\!i}\!\rightarrow\!\bp_{\!s}$ as $m(\{T_{i}\}_{i\in I})\!=\!s[T_{i}]_{i\in I}$. It is easy to verify that $m$ is a homeomorphism. Consider the product measure  $\times_{i\in I}\mathbb{P}_{M_{i}}$ on the space $\prod_{i}\bp_{i}$. We now show that  the equality
\begin{equation}\label{eq_aux_2_fixedpoint_theorem}
\times_{i\in I}\mathbb{P}_{M_{i}}(X)\!=\!  \mathbb{P}_{M}(m(X))
\end{equation}
holds for every measurable $X\!\subseteq\! \prod_{i}\bp_{i}$. By regularity  of measures in Polish spaces (see, e.g., \cite{Kechris}), we just need to prove that for each basic open set $O\!\subseteq\!\prod_{i}\bp_{i}$ the equality $\times_{i\in I}\mathbb{P}_{M^{i}}(O)\!=\!  \mathbb{P}_{M}(m(O))$ holds. The basic sets $O$ in the product topology are of the form $O_{F_{0}}\!\times\! \dots\! \times O_{F_{k}}\times \prod_{i>k}\bp_{i}$ with $O_{F_{n}}\!\subseteq\!\bp_{n}$, for some $k\!\in\!\mathbb{N}$ and $0\leq n\leq k$. As usual, $O_{F_{n}}$ denotes the basic open set of branching plays containing the finite tree $F_{n}$. Equality \ref{eq_aux_2_fixedpoint_theorem} then follows by definition of $m$ and Definition \ref{measure_definition} of the probability measures $\mathbb{P}_{M}$ and $\mathbb{P}_{M_{i}}$. 

The validity of Equation \ref{eq_aux_1_fixedpoint_theorem} is then derived as follows:
\begin{center}
\begin{tabular}{l l l}
$\mathbb{P}_{M}(\Phi_{\pr,\pl})$ & $=_{A}$ & $\mathbb{P}_{M}(\Phi_{\pr,\pl}\cap \bp_{\!s})$\\
$$ & $=_B$ & $\mathbb{P}_{M}\big( m( \overline{\prod_{i\in I}(\overline{\Phi_{\pr,\pl} \cap \bp_{\!i}})}  )\big)$\\
$$ & $=_{C}$ & $1- \mathbb{P}_{M}\big( m( \prod_{i\in I}(\overline{\Phi_{\pr,\pl} \cap \bp_{\!i}})  )\big)$\\
$$ & $=_{D}$ & $1- \displaystyle \times_{i\in I}\mathbb{P}_{M_{i}}(\overline{\Phi_{\pr,\pl} \cap \bp_{\!i}})$\\
$$ & $=_{E}$ & $1- \displaystyle \prod_{i\in I} \big(1- \mathbb{P}_{M_{i}}(\Phi_{\pr,\pl} \cap \bp_{\!i})\big)$\\
$$ & $=_{F}$ & $1- \displaystyle \prod_{i\in I} \big(1- \mathbb{P}_{M_{i}}(\Phi_{\pr,\pl})\big)$\\
$$ & $=_{G}$ & $\displaystyle \coprod_{i\in I}  \mathbb{P}_{M_{i}}(\Phi_{\pr,\pl})$.
\end{tabular}
\end{center}
where the overlined sets denote the expected complements. Steps (A) and (F) hold by previous observations. To justify step (B), observe (using a standard strategy stealing argument) that a branching play $s[T_{i}]_{i\in I}$ is winning for Player $1$ (i.e., it is in $\Phi_{\pr,\pl}$) if and only if at least one of its subtrees $T_{i}$ is winning for Player $1$. This is because the state $s\!\in\!B_{1}$ is under the control of Player $1$ in the inner game. Step $C$ is valid because $m$ is a homeomorphism, thus a bijection. Step $(D)$ holds by Equation \ref{eq_aux_2_fixedpoint_theorem}. Step $E$ holds by definition of product measure. Lastly, step $(G)$ holds by De Morgan dualities of the operations of product and coproduct.
\end{proof}

We conclude this section by remarking that the notion of $2\frac{1}{2}$-player meta-parity game can be further generalized, allowing the inner games to be $2$-player games with general Borel winning sets. We refer to \cite[\S 5]{MioThesis} for an analysis of this sort of games.

\section{Game semantics of pL$\muprod$}
\label{model_checking_games}
In this section we define the game semantics of the logic pL$\mu^{\odot}$. As for L$\mu$ \cite{Stirling96} and pL$\mu$ \cite{MM07, MIO10}, given a PLTS  $\langle P, \{ \freccia{a} \}_{a\in L}\rangle$ and an interpretation of the variables $\rho$  a game $\game^{F}_{\rho}$ is constructed for every pL$\mu^{\odot}$ formula $F$. The game semantics of a formula $F$ at a  state $p$ is defined as the value of the game at a designated state $\langle p, F\rangle$.  The logics L$\mu$ and pL$\mu$ are interpreted using ordinary $2$-player parity games and $2\frac{1}{2}$-player parity games, respectively. As anticipated earlier, we shall interpret pL$\mu^{\odot}$ with the novel class of $2\frac{1}{2}$-player meta-parity games. Following the approach of \cite{Stirling96}, we first identify a class of pL$\mu^{\odot}$ formulas which is easier to work with and allow the simplification of some definitions.

\begin{defi}[\cite{Stirling96}]
\label{normal_form}
We say that a pL$\mu^{\odot}$ formula $F$ is in \emph{normal form}, if every occurrence of a $\mu$ or $\nu$ binder binds a distinct variable, and no 
variable appears both free and bound. Every formula can be put in normal form by standard $\alpha$-renaming of the bound variables.
\end{defi}

\begin{defi}[\cite{Stirling96}] Given a pL$\mu^{\odot}$ formula $F$ in normal form, we say that a variable $X$ \emph{subsumes} a variable $Y$ if $X$ and $Y$ are bound in $F$ by the subformulas $\star_{1}X.G$ and $\star_{2}Y.H$ respectively, and $\star_{2}Y.H\!\in\! Sub(G)$, for $\star_{1},\star_{2}\!\in\!\{\mu,\nu\}$.
\end{defi}

\begin{defi}
We say that a pL$\mu^{\odot}$ formula $F$ is in \emph{product normal form} if 
\begin{enumerate}[(1)]
\item it is in normal form, and
\item it does not contain subformulas of the form $G\cdot G$ or $G \odot G$.  
\end{enumerate}
Every formula can be put in product normal form by (recursively) replacing every subformula $G\star G$ with, e.g., the semantically equivalent formula $G\star (G\vee G)$, for $\star\!\in\!\{ \cdot,\odot \}$.
\end{defi}

In what follows we shall restrict our attention, without loss of generality, to formulas in product normal form. While the advantage of restricting to normal formulas is clear \cite{Stirling96}, the product normal form is useful for the following reason. As anticipated in the introduction, we shall interpret game-states of the form $\langle p, F \cdot G\rangle$ and $\langle p, F \odot G\rangle$ as branching states having the states $\langle p, F\rangle$ and $\langle p, G\rangle$ as children. When $F\!=\!G$, the set of children becomes a singleton, and this does not reflect the intended game interpretation. The product normal form is one of the simplest way  to avoid this kind of situations. As an alternative solution, one might consider \emph{multisets} of successor states, rather than sets, in the definition of $2\frac{1}{2}$-player meta-parity games.

We are now ready to specify how, given a PLTS  $\lts\!=\!\langle P, \{ \freccia{a} \}_{a\in L}\rangle$, the $2\frac{1}{2}$-player meta-parity game $\game^{F}_{\rho}\!=\!\langle \arena, \Phi_{\pr,\pl}\rangle$, associated with a pL$\mu^{\odot}$ formula $F$ and a $[0,1]$-interpretation $\rho$ of the variables, is constructed.

The game arena $\arena=\langle (S,E), (S_{1},S_{2},S_{N},B),\pi  \rangle$ is defined as follows. The set of game states $S$ is defined as $S\!\bydef\! \big(P\!\times\! Sub(F)\big) \cup \big(\mathcal{D}(\lts)\!\times\! Sub(F)\big)\cup \{\top,\bot\}$, where $\mathcal{D}(\lts)$ denotes the (countable) set of probability distribution in $\lts$ (see Definition \ref{PLTS}),  and $\{ \top , \bot\}$ are two special states representing immediate win and loss for Player $1$, respectively. The relation $E$ is defined as follows: $E (\langle d,G \rangle)\!\bydef\!   \{ \langle p, G \rangle \ | \ p\!\in\! \supp(d) \}$ for every $d\!\in\! \mathcal{D}(\lts)$. The set $E(\langle p, G\rangle)$ is defined, by case analysis on the outermost connective of $G$, as follows: 
\begin{enumerate}[(1)]
\item if $G\!=\! X$, with $X$ free in $F$, then $E(\langle p, G\rangle)\!\bydef\! \{ \bot,\top\}$.
\item if $G\!=\! X$, with $X$ bound in $F$ by the subformula $\star X. H$, with $\star \in \{ \mu , \nu \}$, then $E(\langle p, G\rangle)\!\bydef\! \{\langle p, H\rangle\}$.
\item if $G\!=\! \star X.H$, with $\star\! \in\! \{ \mu, \nu \}$, then $E(\langle p,G\rangle)\!\bydef\! \{\langle p, H\rangle\}$.
\item if $G\!=\! \diam{a}H$,  then  $E(\langle p, G\rangle)\!\bydef\!  \{ \langle d, H \rangle\ |\ p\freccia{a}d \}$.
\item if $G\!=\! \quadrato{a}H$,  then $E(\langle p, G\rangle)\!\bydef\!  \{ \langle d, H \rangle\ |\ p\freccia{a}d \}$.
\item if $G\!=\!  H\star H^{\prime}$ with $\star\in \{ \vee,\wedge, \odot, \cdot\}$ then $E(\langle p, G\rangle) \!=\!  \{ \langle p, H \rangle, \langle p,H^{\prime }\rangle \}$
\end{enumerate}
Lastly we define $\top$ and $\bot$ to be terminal states, i.e., $E(\bot)\!=\! E(\top)\!=\! \emptyset$. The partition $( S_{1},S_{2},S_{N}, B )$  of $S$ is defined as follows: every state $\langle p, G \rangle$ with $G$'s main connective in $\{ \diam{a}, \vee, \mu X\}$ or with $G=X$ where $X$ is a $\mu$-variable, is in $S_{1}$. Dually, every state $\langle p, G\rangle$ with $G$'s main connective in $\{ \quadrato{a}, \wedge, \nu X\}$ or with $G=X$ where $X$ is a $\nu$-variable, is in $S_{2}$. Every state of the form $\langle d, G \rangle$ or $\langle p, X\rangle$, with $X$ free in $F$, is in $S_{N}$. Every state $\langle p, G \rangle$ whose $G$'s main connective is either $\cdot$ or $\odot$ is in $B$.  Lastly we define the terminal states $\bot$ and $\top$ to be in $S_{1}$ and $S_{2}$ respectively. The function $\pi\!:\! S_{N}\!\rightarrow \!\mathcal{D}(S)$ assigns a probability distribution to every state under the control of Nature (thus specifying its indended probabilistic behavior) and it is defined as $\pi (\langle d, G\rangle) (\langle p, G\rangle)\!=\!d(p)$ on all states of the form $\langle d, G\rangle$. All other states in $S_{N}$ are of the form $\langle p,X\rangle$, with $X$ free in $F$. The function $\pi$ is defined on these states as follows:
\begin{center}
 $\pi(\langle p,X\rangle)(s)\bydef \left\{     \begin{array}{l  l}  	 \rho(X)(p) & $if $ \  s=\top \\ 
							               			1- \rho(X)(p) & $if $ \  s=\bot \\ 
							               			 0 & $otherwise $  \\  
                      	      		     \end{array}      \right.$
\end{center} 
The priority assignment $\pr\!:\! S\! \rightarrow\! \mathbb{N}$ is defined as usual in $\mu$-calculus games: an odd priority is assigned to the states $\langle p,X\rangle$ with $X$ a $\mu$-variable  and, dually, an even priority is assigned to the states $\langle p,X\rangle$ with $X$ a $\nu$-variable, in such a way that if $Z$ subsumes $Y$ in $F$ then $\pr(\langle p,Z\rangle) > \pr(\langle p, Y\rangle)$. Moreover, for every terminal state $s\!\in\!S$,  we define $\pr(s)\!=\!0$ if $s\!\in\! S_{1}$, and $\pr(s)\!=\!1$ if  $s\!\in\! S_{2}$.  This implements the policy that a player who gets stuck at a terminal state loses (see Theorem \ref{fixed_point_proposition_1}). All other states get assigned, by convention, priority $0$. Lastly, the player assignment $\pl\!:\!  B\!\rightarrow\!\{1,2\}$ is defined as $\pl(\langle p, G_{1}\!\odot\! G_{2}\rangle)\!=\!1$ and $\pl(\langle p, G_{1}\!\cdot\! G_{2}\rangle)\!=\!2$, for every $p\!\in\! P$ and $G_{1},G_{2}\!\in\! Sub(F)$.

\begin{rem}
We now list some useful observations about the above defined game $\game^{F}_{\rho}$.
\begin{enumerate}[(1)]
\item If no (co)product operators occur in $F$, then $B\!=\!\emptyset$, and the game $\game^{F}_{\rho}$ is equivalent to the pL$\mu$ games of \cite{MM07,MIO10} (see Remark \ref{remark_meta_games_without_branching_states}). Thus the game semantics for pL$\mu^{\odot}$ generalizes the game semantics of pL$\mu$, as previous claimed.
\item If the PLTS $\lts$ is finite then the game arena $\arena$ of $\game^{F}_{\rho}$ is finite too.
\item Note how the only role of two game states $\{\top,\bot\}$ is to interpret the game states $\langle p, X\rangle$ with $X$ a free variable. By application of Theorem \ref{fixed_point_proposition_1}, it follows that the game values at $\top$, $\bot$ and $\langle p, X\rangle$ are $1$, $0$, and $\rho(X)(p)$ respectively, as expected.
This solution avoid the  otherwise necessary introduction of $[0,1]$-valued payoff functions in the definition of $2\frac{1}{2}$-player meta-parity games.
\end{enumerate}
\end{rem}

We are now ready to state our main result.
\begin{thm}[\martin]\label{main_theorem}Given a PLTS $\lts\!=\!\langle P, \{ \freccia{a} \}_{a\in L}\rangle$, for every state $p\!\in\! P$, interpretation of the variables $\rho$ and pL$\muprod$ formula $F$, the following equalities hold: 
\begin{center}
$\val{\downarrow}(\game^{F}_{\rho})(\langle p, F\rangle)=\val{\uparrow}(\game^{F}_{\rho})(\langle p, F\rangle)=\sem{F}_{\rho}(p)$.
\end{center}
\end{thm}
Note that Theorem \ref{main_theorem} asserts that pL$\mu^{\odot}$ games are determined. In light of this result, we can just write $\val(\game^{F}_{\rho})$ for the unique value function of the game $\game^{F}_{\rho}$. The \emph{game semantics} of $F$ at $p$ under the interpretation $\rho$ is then defined as $\val(\game^{F}_{\rho})(\langle p, F\rangle)$, and it coincides with the denotational semantics $\sem{F}_{\rho}(p)$, as desired.

The game semantics for pL$\mu^{\odot}$ offers a clear and simple operational interpretation for the meaning of the qualitative threshold modalities $\mathbb{P}_{>0}$ and $\mathbb{P}_{=1}$ (see Definition \ref{threshold_modalities_def}). Let us consider, for example, the  game $\game^{\mathbb{P}_{>0}F}_{\rho}$ associated with a PLTS $\lts\!=\!\langle P, \{ \freccia{a}\}_{a\in L}\rangle$, a $[0,1]$-interpretation $\rho$ and  the pL$\mu^{\odot}$ formula $\mathbb{P}_{>0}F\!\bydef\! \mu X. (F \odot X)$. The game   $\game^{\mathbb{P}_{>0}F}_{\rho}$, at the state $\langle p, \mu X.(F\odot X)\rangle$, for some $p\!\in\! P$, can be depicted as follows:
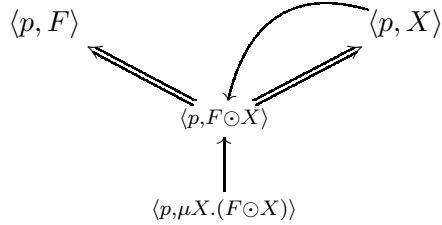
\begin{center}
$$
\SelectTips{cm}{}
	\xymatrix @=20pt {
		\node{\langle p, F\rangle} & & \node{\langle p, X\rangle}  \ar@{->}@/_25pt/[dl] \\
		& \node{\scriptstyle{\langle p, F\odot X\rangle} }
				\ar@{=>}@/_0pt/[ul]
				\ar@{=>}@/_0pt/[ur]
		\\
		&  \node{\scriptstyle{\langle p, \mu X.(F\odot X)\rangle} } \ar@{->}@/_0pt/[u] \\
	}
$$
\end{center}
After an initial unfolding step from $\langle p, \mu X.(F\odot X)\rangle$ to $\langle p, F\odot X\rangle$, the game is split in two concurrent sub-games, one continuing its execution from the state $\langle p, F\rangle$  (this sub-game can be considered an instance of the game $\game^F_\rho$ starting at $\langle p, F\rangle$) and the other from the state $\langle p, X\rangle$.  In order to win  the game $\game^{\mathbb{P}_{>0}F}_{\rho}$, Player $1$ has to win in at least one of  the two generated sub-games, thus either in the instance of $\game^F_\rho$ or in the sub-games continuing at $\langle p, X\rangle$. This second sub-game, however, after an unfolding step, progresses to the game state $\langle p, F\odot X\rangle$, where the protocol is repeated generating yet another two sub-games. The infinite execution of the game leads to the generation of  infinitely many instances of the game $\game^F_\rho$. A branching play $T$ in $\game^{ \mathbb{P}_{>0}F}_{\rho}$ can be depicted as follows: 
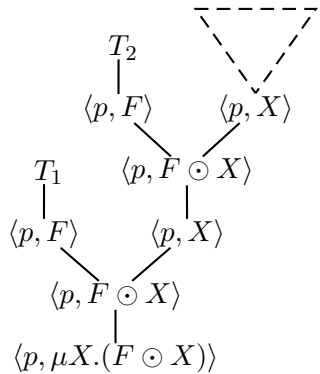
\begin{center}
$\ \ \ \ \ \ \ \ \ $
\pstree[ treemode=U,levelsep=5ex ]{\Tr{$\langle p,\mu X.( F \odot X)\rangle$}}{
\pstree[ treemode=U,levelsep=5ex ]{\Tr{$\langle p, F \odot X\rangle$}}{
 \pstree{\Tr{$\langle p, F\rangle$}}{
  		\TR{ $T_{1} $}
 }
 \pstree{\Tr{$\langle p, X\rangle$}}{
 	\pstree[ treemode=U,levelsep=5ex ]{\Tr{$\langle p, F \odot X\rangle$}}{
	 \pstree{\Tr{$\langle p, F\rangle$}}{
  		\TR{ $T_{2} $}
 }
 	\pstree[ treemode=U,levelsep=8ex,linestyle=dashed ]{\Tr{$\langle p, X \rangle$}}{
        \pstree[linestyle=none,arrows=-,levelsep=2ex]{\Tfan[fansize=10ex]}{\TR{ $\ $}}
	}
	}
}
}
}
\end{center}
where $T_{1}$, $T_{2}$, $\dots$, represent the branching plays corresponding to the plays in each generated instance of the game $\game^F_\rho$. Since the variable $X$ unfolded infinitely often in the rightmost path in $T$ is bound by a least fixed point in $\mu X.(F\odot X)$, and since the $\odot$ nodes are Player $1$ choices in the inner game $\game_{T}$, we have that $T$ is a winning branching play for Player $1$ if and only if there exists some $n\!\in\!\mathbb{N}$ such that $T_{n}$, the outcome of the $n$-th generated instance of $\game(F,\rho)$, is winning for Player $1$.  Thus, the game $\game^{\mathbb{P}_{>0}F}_{\rho}$, at the state $\langle p, \mu X.F\odot X\rangle$ can be simply described as follows: generate an infinite number of instances of the game $\game^F_\rho$ at the state $\langle p, F\rangle$; Player $1$ wins if at least one of the infinitely many generated instances of $\game^F_\rho$ ends up in a winning branching play and Player $2$ wins otherwise. It is then quite clear that if $\val(\game^{\mathbb{P}_{>0}F}_{\rho})(\langle p, \mathbb{P}_{>0}F \rangle)\!>\!0$ (or, equivalently, $\sem{F}_{\rho}(p)\!>\! 0$ by  Theorem \ref{main_theorem}), then the probability that at least one (and in fact countably many) of the infinitely many instances of $\game^F_\rho$ will result in a win for Player $1$, is $1$. Similarly, if $\sem{F}_{\rho}(p)\!=\!0$, then  the probability that at least one of the instances of $\game^F_\rho$ will result in a win for Player $1$, is $0$.  

The game semantics for pL$\mu^{\odot}$ thus offers a straightforward interpretation for the probabilistic qualitative modality $\mathbb{P}_{>0}$ exploiting the simple idea that an event (which we can, at some extent, see as  \emph{a} pL$\mu^{\odot}$ property) has probability greater than zero if and only if, when repeated infinitely many time, it almost surely occurs at least once. An analogous straightforward interpretation can be given to the other qualitative threshold modality  $\mathbb{P}_{=1}F\!\bydef\! \nu X. (F\cdot X)$: generate an infinite number of instances of the game $\game^F_\rho$ at the state $\langle p, F\rangle$; Player $1$ wins if all of them  end up in a winning branching play for Player $1$, and Player $2$ wins otherwise.


\section{Inductive characterization of the winning set of $\game^{\mu X.F}_{\rho}$}
\label{technical_section}
In this section we provide a transfinite inductive characterization of the  set $\Phi^{\mu X.F}$ of winning branching plays of the game $\game^{\mu X.F}_{\rho}$, for an arbitrary pL$\mu^{\odot}$ formula $F$ and interpretation $\rho$. This result will be  used in the proof of our main result, in Section \ref{proof_section}. The inductive characterization of $\Phi^{\mu X.F}$ will be obtained exploiting the similarities between the pL$\mu^{\odot}$ games $\game^{\mu X.F}_{\rho}$ and the simpler (in the complexity of the formula) game $\game^{F}_{\rho}$. The game arenas of the two games, denoted here by $\arena^{\mu X.F}$ and $\arena^{F}$, are almost identical  as they differ only in the following aspects:
\begin{enumerate}[(1)]
\item The set of game states  of $\arena^{\mu X.F}$ is the set of states of $\arena^{F}$  plus the set of states of the form $\langle p, \mu X.F\rangle$. The latter states, however, play almost no role in the game because they have only one successor, namely $\langle p, F\rangle$, and are not reachable by any other state.
\item More importantly, the states of the form $\langle p, X\rangle$, which are present in both game arenas, are  Player $1$ states in $G^{\mu X.F}_{\rho}$ (with $\langle p, F\rangle$  as unique successor), and  Nature states in $\game^{F}_{\rho}$ (with two\footnote{\label{nota_successors}Note that the set of successor state of $\langle p, X\rangle)$ is defined to  be $\{\top,\bot\}$, even when $\rho(X)(p)\!\in\{0,1\}$.} \emph{terminal} successors, $\top$ and $\bot$, reachable with probability $\rho(X)(p)$ and $1\!-\! \rho(X)(p)$, respectively).
\end{enumerate}
Moreover observe that the player assignments of the two games, denoted here by $\pl^{F}$ and $\pl^{\mu X.F}$, are identical, and that the priority assignments, $\pr^{F}$ and $ \pr^{\mu X.F}$, differ only on the game-states $s$ of the form $\langle p, X\rangle$:  $\pr^{F}(s)\!=\!0$ and $\pr^{\mu X.F}(s)\!=\!p$ for some odd (maximal) priority $p\!=\!\max (\pr^{\mu X.F})$ (see Definition \ref{parity_assignment_def}).

A branching play $T$ in the game $\game^{F}_{\rho}$, rooted at the game state $\langle p, F\rangle$, can be depicted as in figure \ref{fig1}. The triangle represents\footnote{The picture is quite simplistic. For example there are, in general, paths in $T$ never reaching a state of the form $\langle q, X\rangle$, that nonetheless branch away from the paths $\vec{s}_{i}$, for $i\!\in\!I$, somewhere between the root $\langle p, F\rangle$ and the last state $\langle p_{i},X\rangle$, whereas the picture depicts all such branches as branching away immediately at the root $\langle p, F\rangle$.} the set of paths in $T$ never reaching a state of the form $\langle q, X\rangle$, for $q\in P$, and the other edges represents the (possibly empty) collection of paths $\{\vec{s}_{i}\}_{i\in I}$ in $T$, for some countable index set $I$, reaching a state of the form $\langle p_{i}, X\rangle$ which is (necessarily) followed by a terminal state $b_{i}\in \{\top,\bot\}$.  Similarly a  branching play $T$ in $\game_{\rho}^{\mu X.F}$, rooted at $\langle p, F\rangle$, can be depicted as in figure \ref{fig2}. We extract the common part between the branching plays of $\game^{F}_{\rho}$ and $\game^{\mu X.F}_{\rho}$ by the following definition.

\begin{defi}[Branching pre-play]
Let $T$ be a branching play in $\game^{F}_{\rho}$ or in $\game^{\mu X.F}_{\rho}$, and $\{\vec{s}_{i}\}_{i\in I}$ be the $I$-indexed set of paths in $T$ reaching a state of the form $\langle p_{i}, X\rangle$, as described above. The \emph{branching pre-play} $T[x_{i}]_{i\in I}$, which can be depicted as in figure \ref{fig3}, is the tree obtained from the branching play $T$ by pruning its subtrees rooted at $\vec{s}_{i}$, for $i\!\in\!I$. \end{defi}
\begin{figure}[t]
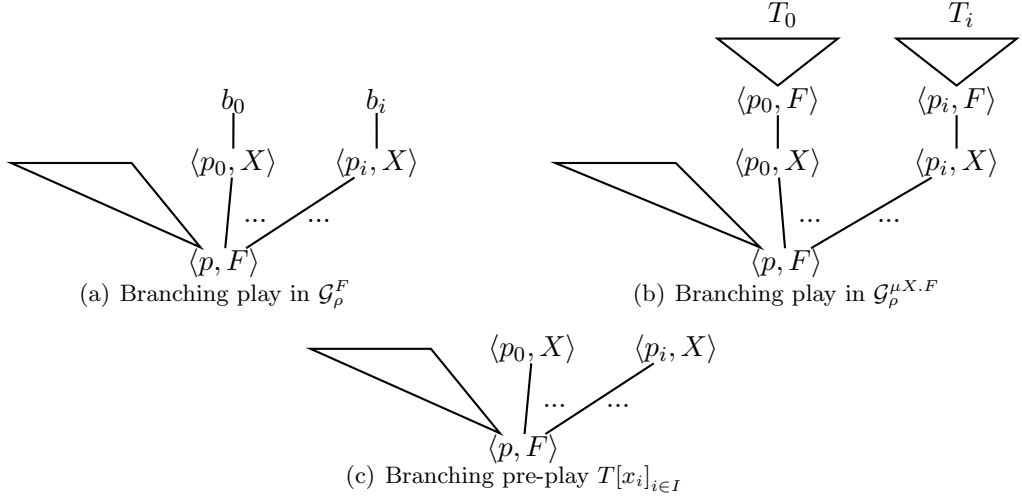

\centering
\subfigure[Branching play in  $\game^{F}_{\rho}$]{\label{fig1}
		\pstree[ treemode=U,levelsep=8ex ]{\Tr{$\langle p, F\rangle$}}{
        		\pstree[linestyle=none,arrows=-,levelsep=2ex]{\Tfan[fansize=10ex]}{\TR{ $\ $}}
		\pstree[levelsep=5ex]{\Tr{$\langle p_{0},X\rangle$} \trput{$...$}}{ \Tr{$ b_{0} $} } 
		\pstree[levelsep=5ex]{\Tr{$\langle p_{i},X\rangle$} \trput{$...$}}{ \Tr{$ b_{i} $} } 
		}
}\qquad\qquad 
\subfigure[Branching play in  $\game^{\mu X.F}_{\rho}$]{\label{fig2}
	\pstree[ treemode=U,levelsep=8ex ]{\Tr{$\langle p, F\rangle$}}{
        \pstree[linestyle=none,arrows=-,levelsep=2ex]{\Tfan[fansize=10ex]}{\TR{ $\ $}}
  
\pstree[levelsep=5ex]{\Tr{$\langle p_{0},X\rangle$} \trput{$...$}}{
	\pstree[levelsep=5ex]{\Tr{$\langle p_{0},F\rangle$}}{
		\pstree[linestyle=none,arrows=-,levelsep=2ex]{\Tfan[fansize=10ex]}{\TR{ $T_{0} $}}
	
	}
}
\pstree[levelsep=5ex]{\Tr{$\langle p_{i},X\rangle$} \trput{$...$}}{
	\pstree[levelsep=5ex]{\Tr{$\langle p_{i},F\rangle$}}{
		\pstree[linestyle=none,arrows=-,levelsep=2ex]{\Tfan[fansize=10ex]}{\TR{ $T_{i} $}}
	}
}
}		
}
\subfigure[Branching pre-play $T{[}x_{i}{]}_{i\in I}$]{\label{fig3}
	\pstree[ treemode=U,levelsep=8ex, ]{\Tr{$\langle p, F\rangle$}}{
        \pstree[linestyle=none,arrows=-,levelsep=2ex]{\Tfan[fansize=10ex]}{\TR{ $\ $}}

        \Tr{$\langle p_{0},X\rangle$}
        \trput{$...$}
	\Tr{$\langle p_{i}, X\rangle$}		
        \trput{$...$}
	}
}
\caption{Branching plays and pre-plays}
\end{figure} 

The notation adopted for branching pre-plays is motivated by the use we make of them. We consider a branching pre-play $T[x_{i}]_{i\in I}$ as a context on which we can plug in other branching plays: at the hole indexed by $i$, for $i\!\in\! I$, any branching play rooted at $\langle p_{i}, X\rangle$ can be plugged. If we desire to obtain a branching play for the game $\game^{F}_{\rho}$, we shall fill the branching pre-play $T[x_{i}]_{i\in I}$ with branching plays $T_{i}$ rooted at $\langle p_{i},X\rangle$ in $\arena^{F}$. These are  trees ending immediately either in  the leaf $\top$ or in the leaf $\bot$. We denote with $T[b_{i}]_{i\in I}$ the branching play in $\game^{F}_{\rho}$ obtained by filling the hole indexed by $i$ with the branching play rooted at $\langle p_{i},X\rangle$ and having $b_{i}$ as leaf, for $b_{i}\!\in\!\!\{\top,\bot\}$ (see Figure \ref{fig1}). Similarly given a  $I$-indexed family $\{T_{i}\}_{i\in I}$ of branching plays in $\game^{\mu X.F}$, where each $T_{i}$ is rooted in $\langle p_{i},X\rangle$,  we denote with $T[T_{i}]_{i\in I}$ the branching play in $\game^{\mu X.F}_{\rho}$ obtained by filling the hole indexed by $i$ with the branching play  $T_{i}$ (see Figure \ref{fig2}).

Clearly every branching play $T$ in $\game^{F}_{\rho}$, rooted at $\langle p, F\rangle$, is uniquely of the form $T[b_{i}]_{i\in I}$ for an appropriate sequence $\{b_{i}\}_{i\in I}$.  Similarly every branching play $T$  rooted at $\langle p, F\rangle$ in $\game^{\mu X.F}_{\rho}$ is of the form $T[T_{i}]_{i\in I}$ for an appropriate sequence $\{T_{i}\}_{i\in I}$. We now exploit these observations to define the following function from branching plays in $\game^{\mu X.F}_{\rho}$ to branching plays in $\game^{F}_{\rho}$.

\begin{defi}\label{function_m_X}
Let $X\!\subseteq\! \bp^{\mu X.F}$ be a set of branching plays in the game $\game^{\mu X.F}_{\rho}$. We define the function $m_{X}\!:\! \bp^{\mu X.F}\!\rightarrow\! \bp^{F}$, from branching plays in $\game^{\mu X.F}_{\rho}$ to branching plays in $\game^{F}_{\rho}$, as follows:
\begin{center}
$m_{X}( T[T_{i}]_{i\in I}) =  T[T_{i} \underline{\in} X ]_{i\in I}\ \ \ \ \ $ where $T_{i} \underline{\in} X \bydef  \left\{     \begin{array}{l  l}  	 \top & $if $ \  T_{i}\in X \\ 
							               			\bot & $otherwise $  \\  
                      	      		     \end{array}      \right.$

\end{center}
\end{defi}
Thus, the function $m_{X}$ maps a branching play (uniquely expressible as $T[T_{i}]_{i\in I}$, as observed above) to the corresponding branching play $T[b_{i}]_{i\in I}$ in $\game^{F}_{\rho}$, obtained by filling the $i$-th hole of $T[x_{i}]_{i\in I}$, with the branching play having $\top$ as leaf if and only if $T_{i}$ belongs to the set $X$.

 \begin{lem}\label{m_X_univ_meas}
 Let $X\!\subseteq\! \bp^{\mu X.F}$ be an open (Borel, universally measurable) set of branching plays in the game $\game^{\mu X.F}_{\rho}$. Then the function $m_{X}$ is continuous (Borel measurable, universally measurable).
 \end{lem}

 \begin{proof}
It is clear that the function $\lambda T. ( T \underline{\in } X)$ of Definition \ref{function_m_X} is continuous (Borel, universally measurable) precisely when $X$ is open (Borel, universally measurable), where $\{\top,\bot\}$ is endowed with the discrete topology.
The proof is then carried out with standard arguments (see, e.g., Theorem 6.2.18 in \cite{MioThesis}). Here we omit the routine details.
 \end{proof}
 
We are now ready to define the operator of which $\Phi^{\mu X.F}$, the winning set of the game $\game^{\mu X.F}_{\rho}$, is the least fixed point.

\begin{defi}\label{def_mathbb_W}
The operator $\mathbb{W}\!: \mathcal{\bp}^{\mu X.F}\!\rightarrow\! \mathcal{\bp}^{\mu X.F}$ is defined as follows:
\begin{center}
$\mathbb{W}(X)\bydef m_{X}^{-1}(\Phi^{F})= \{ T[ T_{i}]_{i\in I} \ | \  T[ T_{i} \underline{\in} X ]_{i\in I} \in  \Phi^{F}\}$
\end{center}
where $\Phi^{F}$ is the winning set of the game $\game^{F}_{\rho}$.
\end{defi}

We now prove a few important properties of the operator $\mathbb{W}$.

\begin{lem}\label{winning_preserve_lemma}
If $X$ is a set of branching plays winning for Player $1$ in $\game^{\mu X.F}_{\rho}$, i.e., if $X\!\subseteq\!\Phi^{\mu X.F}$, then $\mathbb{W}(X)\!\subseteq\!\Phi^{\mu X.F}$ too.
\end{lem}

\begin{proof}
Fix some $X\!\subseteq\!\Phi^{\mu X.F}$ and consider an arbitrary branching play $T_{1}=T[T_{i}]\!\in\!\mathbb{W}(X)$. It follows by definition of $\mathbb{W}$ that $T_{2}=T[T_{i}\underline{\in} X]_{i\in I}\!\in\! \Phi^{F}$, i.e., $T_{2}$ is a winning branching play in the game $\game^{F}_{\rho}$. Equivalently, Player $1$ has a winning strategy $\sigma$ in the inner game $G_{T_{2}}$ (see Section \ref{section_meta_parity}). We now prove that $T_{1}\!\in\!\Phi^{\mu X.F}$, i.e., that Player $1$ has a winning strategy $\tau$ in the inner game $G_{T_{1}}$ (see Section \ref{section_meta_parity}) by a strategy stealing argument, exploiting the common structure (the branching pre-play $T[x_{i}]_{i\in I}$) of the two branching plays $T_{1}$ and $T_{2}$. The strategy $\tau$ behaves as the strategy $\sigma$ until a hole $\vec{s}_{i}$, for $i\!\in\! I$, is reached. Thus if no hole is ever reached, the plays in the two games are identical, hence Player $1$ wins following $\tau$, as desired. If a hole $\vec{s}_{i}\!\in\! I$ is reached, then it is necessarily the case that $T_{i}\!\in\! X$, because otherwise the play in $G_{T_{2}}$ would end in the losing state $\bot\!=\! T_{i}\underline{\in}X$ while playing in accordance with the winning strategy $\sigma$. A contradiction. We then define the strategy $\tau$ to play the rest of the game as a winning strategy $\tau_{i}$ for the inner game $\game_{T_{i}}$. Note that such a strategy exist because $T_{i}\!\in\! X\subseteq \Phi^{\mu X.F}$. It then follows that Player $1$, playing in accordance with $\tau$, always produces a play with a winning tail. Thus, $\tau$ is winning, as desired.
\end{proof}

\begin{lem}
The operator $\mathbb{W}$ is monotone, i.e., $\mathbb{W}(X)\!\subseteq\!\mathbb{W}(Y)$ holds for every  $X\!\subseteq\! Y$.
\end{lem}

\begin{proof}
Fix $X\!\subseteq\! Y\!\subseteq \!\bp^{\mu X.F}$. Assume $T_{1}\!=\!T[T_{i}]_{i\in I}\!\in\! \mathbb{W}(X)$, i.e.,  $T[ T_{i} \underline{\in} X ]_{i\in I} \in  \Phi^{F}$. We need to prove that $T_{1}\!\in\! \mathbb{W}(Y)$ too, i.e., that $T_{2}=T[ T_{i} \underline{\in} Y ]_{i\in I} \in  \Phi^{F}$. Equivalently, we need to show that if Player $1$ has a winning strategy in the inner game $G_{T_{1}}$ associated with $T_{1}$, then they have a winning strategy in $G_{T_{2}}$ too. The two parity games  $G_{T_{1}}$ and $G_{T_{2}}$ are almost identical, except that a play ending in one of the holes $\vec{s}_{i}$, for $i\!\in\! I$,  might be losing for Player $1$ in $G_{T_{1}}$ (when the game ends in the leaf $\bot\!=\! T_{i}\underline{\in}X)$ but winning in $G_{T_{2}}$ (when  $T_{i}\underline{\in}Y\!=\!\top)$. The desired result then trivially follows. 
\end{proof}

As a consequence, by application of the Knaster--Tarski theorem, the operator $\mathbb{W}$ has a least fixed point $\lfp(\mathbb{W})$. We are now ready to prove the main result of this section.

\begin{thm}
\label{inductive_ch}
The following equality holds: $\Phi^{\mu X.F}= \lfp( \mathbb{W})$.
\end{thm}

\begin{proof}
We already know, by application of Lemma \ref{winning_preserve_lemma},  that $\lfp(\mathbb{W})\!\subseteq\!\Phi^{\mu X.F}$. We now prove that the equality holds by showing that, for every $T\not\in \lfp(\mathbb{W})$, the branching play $T$ does not belong to  $\Phi^{\mu X.F}$. Fix an arbitrary $T\!\not\in\!  \lfp(\mathbb{W})$. We show that $T\!\not\in\! \Phi^{\mu X.F}$ by constructing a winning strategy $\sigma^{T}_{2}$ for Player $2$ in the inner game $G_{T}$.
By assumption we have that $T=T[T_{i}]_{i\in I}\!\not\in\! \lfp(\mathbb{W})$ or, equivalently, $T[T_{i}\underline{\in}\lfp(\mathbb{W})]_{i\in I}\!\not\in\! \Phi^{F}$. For notational convenience, let us denote with $R$ the branching play $T[T_{i}\underline{\in}\lfp(\mathbb{W})]_{i\in I}$ in $\game^{F}_{\rho}$. Let $\tau^{R}_{2}$ be a strategy winning for Player $2$ in the inner game $G_{R}$. As already done in the proof of Lemma \ref{winning_preserve_lemma}, we define $\sigma^{T}_{2}$ exploiting the common structure (the branching pre-play $T[x_{i}]_{i\in I}$) of $T$ and $R$. The strategy $\sigma^{T}_{2}$  behaves as the strategy $\tau^{R}_{2}$ until a hole $\vec{s}_{i}$, for $i\!\in\! I$, is reached. Thus if no hole is ever reached, the plays in the two games are identical, hence Player $2$ wins following $\sigma^{T}_{2}$, as desired. If a hole $\vec{s}_{i}\!\in\! I$ is reached, then it is necessarily the case that $T_{i}\!\not\in\! \lfp(\mathbb{W})$, because otherwise the play in $G_{R}$ would end in the  state $\top\!=\! T_{i}\underline{\in}\lfp(\mathbb{W})$, which is winning for Player $1$, while playing in accordance with the winning strategy $\sigma^{R}_{2}$. A contradiction. We then define the strategy $\sigma^{T}_{2}$ to play the rest of the game (forgetting the previous history) as the strategy $\sigma^{T_{i}}_{2}$, constructed as for $\sigma^{T}_{2}$, but with respect to the branching play $T_{i}$. We shall call this a \emph{re-starting} point of the play.
We now prove that the strategy $\sigma^{T}_{2}$ is winning for Player $2$ as desired. There are two cases to consider. A play in the game $\game_{T}$ played in accordance with $\sigma^{T}_{2}$ can have,
\begin{enumerate}[(1)]
\item either finitely many restarting points, or
\item infinitely many restarting points, i.e., infinitely many occurrences of states of the form $\langle p_{i},X\rangle$, for $i\!\in\! I$.
\end{enumerate}
In the first case,  following earlier observations, the resulting path has a winning tail and thus is winning. In the second case, since the states of the form $\langle p_{i},X\rangle$ are assigned maximal odd priority in $\game^{\mu X.F}_{\rho}$, the play is winning for Player $2$ (see Definition \ref{WPR}).
\end{proof}

The following results will be used for dealing with the measure-theoretic complications associated with the complexity of the winning set $\Phi^{\mu X.F}$.

\begin{lem}\label{iteration_lemma}
The least fixed point of $\mathbb{W}$ is reached in at most $\omega_{1}$ iterations, i.e., the equality $\lfp(\mathbb{W})\!=\!\bigcup_{\alpha<\omega_{1}} \!\mathbb{W}_{\alpha}$ holds, where $\mathbb{W}_{\alpha}\!=\!\bigcup_{\beta<\alpha}\mathbb{W}(\mathbb{W}_{\beta})$.
\end{lem}
\begin{proof}
Assume $T[T_{i}]_{i\in I}\!\in\! \mathbb{W}_{\omega_{1}+1}$ or, equivalently, $T[T_{i}\underline{\in}   \mathbb{W}_{\omega_{1}}]\!\in\!\Phi^{F}$. 
Let $J=\{ i \ |  \ T_{i}\!\in\! \mathbb{W}_{\omega_{1}} \}$.
For each $j\!\in\! J$, let $\beta_{j}$ be the least ordinal such that $T_{j}\!\in\! \mathbb{W}_{\beta_{j}}$. Note that $\beta_{j}$ is a countable ordinal, since $\mathbb{W}_{\omega_{1}}= \bigcup_{\alpha<\omega_{1}}\mathbb{W}_{\alpha}$, and $\omega_{1}$ is the least uncountable ordinal. Let $\beta=\bigsqcup_{j\in J}\beta_{j}$ be the supremum ordinal of all $\beta_{j}$. Note that $\beta$ is a  countable ordinal  since $I$ is countable and $J\!\subseteq\!I$. It then follows, by Definition \ref{def_mathbb_W}, that $T[T_{i}\underline{\in}   \mathbb{W}_{\beta}]$ and $T[T_{i}\underline{\in}   \mathbb{W}_{\omega_{1}}]$ are identical branching plays. Thus $T[T_{i}]_{i\in I}\!\in\! \mathbb{W}_{\beta+1}\subseteq \mathbb{W}_{\omega_{1}}$. Hence $\mathbb{W}_{\omega_{1}+1}\subseteq \mathbb{W}_{\omega_{1}}$ as desired.
\end{proof}

The result of Lemma \ref{iteration_lemma} can be shown to be strict. One can indeed construct a pL$\mu^{\odot}$ game $\game^{F}_{\rho}$ such that $\mathbb{W}_{\alpha}\!\subsetneq\!\Phi^{\mu X.F}$ for every countable ordinal $\alpha$. We refer to Lemma 6.3.8 in \cite{MioThesis} for a proof of this fact.

We now show that each set in the chain $\{\mathbb{W}_{\alpha}\}_{\alpha<\omega_{1}}$, having $\Phi^{\mu X.F}$ as limit, is provably universally measurable in $\textnormal{ZFC}+\textnormal{MA}_{\aleph_{1}}$ set theory.

\begin{lem}[\martin]\label{W_is_um}
If $X\!\subseteq\bp^{\mu X.F}$ is universally measurable then so is $\mathbb{W}(X)$, i.e., the function $\mathbb{W}$ maps universally measurable sets to universally measurable sets.
\end{lem}

\begin{proof}
By Definition \ref{def_mathbb_W}, we need to show that $\mathbb{W}(X)\bydef m_{X}^{-1}(\Phi^{F})$ is universally measurable. By application of Theorem \ref{delta_complexity} and Theorem \ref{consequences_martin}, the winning set $\Phi^{F}$ is provably universally measurable in $\textnormal{ZFC}+\textnormal{MA}_{\aleph_{1}}$. The desired result then follows by application of Lemma \ref{m_X_univ_meas} and Theorem \ref{properties_univ_meas}(b).
\end{proof}

The techniques adopted in this section can be trivially adapted to get the expected dual results which we simply summarize in the following proposition.
\begin{prop}\label{dual_results}
The following assertion holds for every pL$\mu^{\odot}$ formula $\nu X.F$ and interpretation $\rho$:
\begin{enumerate}[\em(1)]
\item the winning set $\Phi^{\nu X.F}$ is the greatest fixed point of the operator $\mathbb{W}$,
\item $\gfp(\mathbb{W})=\bigcap_{\alpha<\omega_{1}}\mathbb{W}^{\alpha}$, where $\mathbb{W}^{\alpha}=\bigcap_{\beta<\alpha}\mathbb{W}(\mathbb{W}^{\beta})$,
\item (\martin) for every countable ordinal $\alpha$, the set $\mathbb{W}^{\alpha}$ is universally measurable.
\end{enumerate}
\end{prop}

\section{Robust Markov Branching Plays}\label{robust_markov_plays_section}
In this section we identify a property of Markov branching plays in pL$\mu^{\odot}$ games which will be used in the proof of Theorem \ref{main_theorem} in the next section. 

Given a pL$\muprod$ game $G^{F}_{\rho}$, with $F$ a pL$\mu^{\odot}$ formula (having a free variable $X$) and $\rho$ a $[0,1]$-interpretation of the variables, respectively, we extend the graphical notation, introduced in  Section \ref{technical_section} for branching plays, to Markov branching plays in the expected way. Thus, given a Markov branching play $M$ in $G^{F}_{\rho}$ rooted at $\langle p, F\rangle$ (for some state $p$), we depict $M$  as in Figure \ref{figB1-A} and denote it by $M[\lambda_{i}]_{i\in I}$, with $\lambda_{i}\!\in\![0,1]$ being the probability labeling the edge connecting $\langle p_{i},X\rangle$ with $\top$.
\begin{figure}
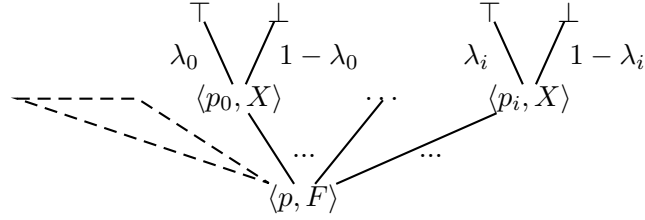

	\pstree[ treemode=U,levelsep=8ex ]{\Tr{$\langle p, F\rangle$}}{
        \pstree[linestyle=none,arrows=-,levelsep=2ex]{\Tfan[fansize=10ex, linestyle=dashed]}{\TR{ $\ $}}
  
\pstree[levelsep=7ex]{\Tr{$\langle p_{0},X\rangle$} \trput{$...$}}{
	\Tr{$\top$}\tlput{$\lambda_{0}$}
	\Tr{$\bot$}\trput{$1-\lambda_{0}$}
}
\pstree[levelsep=7ex, linestyle=none]{\Tr{$\dots $}}{
	\Tr{$\ $}
	\Tr{$\ $}
}
\pstree[levelsep=7ex]{\Tr{$\langle p_{i},X\rangle$} \trput{$...$}}{
	\Tr{$\top$}\tlput{$\lambda_{i}$}
	\Tr{$\bot$}\trput{$1-\lambda_{i}$}
}
}	
\caption{Markov branching play $M{[}\lambda_{i}{]_{i \in I}}$}
\label{figB1-A}
\end{figure}
Note that, by definition of the game $\game^{F}_{\rho}$, the value $\lambda_{i}$ coincides with $\rho(X)(p_{i})$, for $i\!\in\!I$. However, it is convenient to consider, as a technical tool, Markov branching plays of the form $M[\lambda_{i}]_{i\in I}$ having $\lambda_{i}$,  for $i\!\in\!I$, of an arbitrary value, even though these plays never really correspond to achievable plays in $\game^{F}_{\rho}$. The associated probability measure $\mathbb{P}_{M[\lambda_{i}]_{i\in I}}$ over branching plays in $\game^{F}_{\rho}$ is defined as expected.

\begin{defi}[Robust Markov branching plays]
\label{robust}
Fix a pL$\muprod$ formula $F$ (with a free variable $X$) and a $[0,1]$-interpretation $\rho$. We say that a Markov branching play $M[\lambda_{i}]_{i\in I}$ in $\game^{F}_{\rho}$, for $I\subseteq \mathbb{N}$, is \emph{robust in the variable $X$}, or just \emph{robust} when $X$ is clear from the context, if for every $\varepsilon>0$ the following properties hold,
\begin{enumerate}[(1)]
\item $\expected(M[\gamma_{i}]) \geq \expected(M[\lambda_{i}]) - \sum_{i \in I} \frac{\varepsilon}{2^{i+1}}$, and
\item $\expected(M[\delta_{i}]) \leq \expected(M[\lambda_{i}]) + \sum_{i \in I} \frac{\varepsilon}{2^{i+1}}$, 
\end{enumerate}
for every sequences $\{\gamma_{i}\}_{i\in I}$ and $\{\delta_{i}\}_{i \in I}$ of reals in $[0,1]$ such that, for every $i\!\in\! I$, the inequalities $\gamma_{i}\geq \lambda_{i}-\frac{\varepsilon}{\#(i)}$ and $\delta_{i}\leq \lambda_{i}+\frac{\varepsilon}{\#(i)}$ hold (see Definition \ref{approx_function} of $\#\!:\!\mathbb{N}\!\rightarrow\!\mathbb{N}$).
\end{defi}

The notion of robustness captures a useful substitutivity property. If, in a Markov branching play $M[\lambda_{i}]_{i\in I}$, the probability $\lambda_{i}$ of reaching from the state $\langle p_{i},X\rangle$ the winning (for Player $1$) state $\top$ is replaced with a smaller but close enough value $\gamma_{i}$, then the resulting Markov branching play $M[\gamma_{i}]_{i\in I}$ has an expected value close to that of $M[\lambda_{i}]_{i\in I}$ too. 

Note how, in Definition \ref{robust}, the constraint on the distance between $\gamma_{i}$ ($\delta_{i})$ and $\lambda_{i}$ crucially depends on the index $i\!\in\! I$. Definition \ref{robust} has been identified to meet the nature of Markov branching plays in pL$\mu^{\odot}$ games and, as we shall see in the next section, every Markov branching play in a pL$\mu^{\odot}$ game is indeed robust in every free variable. However, it is useful to observe that the dependence on $i\!\in\! I$ for the constraint on the distance between $\gamma_{i}$ ($\delta_{i}$) and $\lambda_{i}$ is necessary. Indeed, the simpler property
\begin{equation}\label{naive_candidate}
(2^{\prime}) \ \ \ \  \ \expected(M[\delta_{i}]) \leq \expected(M[\lambda_{i}]) + c_{1} \ \ \ \textnormal{ where } \forall i . \big( \delta_{i} < \lambda_{i} + c_{2}\big)
\end{equation}
for some constant values $c_{1},c_{2}\!\in\! (0,1)$, is not satisfied (in general) by pL$\mu^{\odot}$ Markov branching plays. In other words, Markov branching plays in pL$\mu^{\odot}$ are not stable (in their expected value) if the values $\lambda_{i}$ are altered uniformly, i.e., by a fixed $c_{2}\!>\!0$. Consider, for example, the Markov branching play, having countably many holes, depicted as in Figure \ref{figura_illustrative}. Assume that $\lambda_{i}\!=\!0$, for every $i\!\in\!\mathbb{N}$ and fix some $c_{2}\!>\!0$.  Then it is simple to verify that $\expected(M[\lambda_{i}])\!=\!0$ and $\expected(M[\delta_{i}])\!=\!1$ (where $\delta_{i}\!=\!c_{2}$, for all $i\!\in\! I$), contradicting ($2^{\prime}$) above. This phenomenon reflects the discontinuity of the denotational interpretation of pL$\mu^{\odot}$ formulas on the free variables (see Proposition \ref{not_continuous_proposition}). Indeed note that the play of Figure \ref{figura_illustrative} is a Markov branching play of a pL$\mu^{\odot}$ game associated with a formula of the form $\mu X.(Y\odot X)\bydef \mathbb{P}_{>0}Y$, similar to the one discussed in the proof of Proposition \ref{not_continuous_proposition}.

\begin{figure}
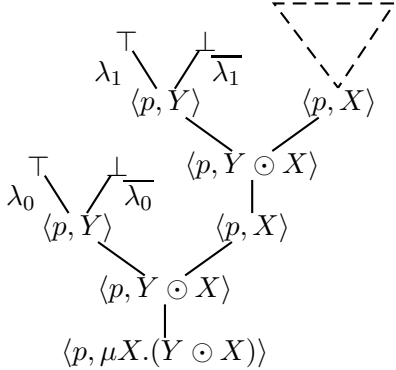

\pstree[ treemode=U,levelsep=5ex ]{\Tr{$\langle p, \mu X.(Y \odot X)\rangle$}}{
\pstree[ treemode=U,levelsep=5ex ]{\Tr{$\langle p, Y \odot X\rangle$}}{
 \pstree{\Tr{$\langle p, Y\rangle$}}{
	\Tr{$\top$}\tlput{$\lambda_{0}$}
	\Tr{$\bot$}\trput{$\overline{\lambda_{0}}$}
 }
 \pstree{\Tr{$\langle p, X\rangle$}}{
 	\pstree[ treemode=U,levelsep=5ex ]{\Tr{$\langle p, Y \odot X\rangle$}}{
	 \pstree{\Tr{$\langle p, Y\rangle$}}{
	\Tr{$\top$}\tlput{$\lambda_{1}$}
	\Tr{$\bot$}\trput{$\overline{\lambda_{1}}$}
 }\pstree[ treemode=U,levelsep=8ex,linestyle=dashed ]{\Tr{$\langle p, X \rangle$}}{
        \pstree[linestyle=none,arrows=-,levelsep=2ex]{\Tfan[fansize=10ex]}{\TR{ $\ $}}{}
	}
	}
}
}
}
\caption{Illustrative example. The symbol $\overline{\lambda_{i}}$ denotes the value $1-\lambda_{i}$.}
\label{figura_illustrative}
\end{figure}

We now establish a useful property relating expected values of Markov branching plays in $\game^{G}_{\rho}$ and in the game $\game^{\mu X.G}_{\rho}$. Again, following established notation, we denote with $M[M_{i}]_{i\in I}$ the Markov branching play depicted as in Figure \ref{figB1}. The \emph{Markov branching pre-play} $M[x_{i}]_{i\in I}$ captures the common structure of the Markov branching plays $M[M_{i}]_{i\in I}$ and $M[\lambda_{i}]_{i\in I}$ in $\game^{\mu X.F}_{\rho}$ and $\game^{F}_{\rho}$, respectively.

\begin{lem}\label{comparing_lemma}
Let $\mu X.F$ be a pL$\mu^{\odot}$ formula and $\rho$ a $[0,1]$-interpretation of the free variables. Let $M[M_{i}]_{i\in I}$ be a Markov branching play in $\game^{\mu X.F}_{\rho}$. For an ordinal $\beta$, let $\gamma^{\beta}_{i}\!=\! \mathbb{P}_{M_{i}}(\bigcup_{\alpha<\beta}\mathbb{W}_{\alpha})$ be the probability of the event  $\bigcup_{\alpha<\beta}\mathbb{W}_{\alpha}$ (see Lemma \ref{iteration_lemma}) associated with the sub-Markov branching play $M_{i}$, for $i\!\in\!I$. Then the  equality 
\begin{center}
$\mathbb{P}_{M[M_{i}]_{i\in I}}(\mathbb{W}_{\beta}) =  \mathbb{P}_{M[\gamma^{\beta}_{i}]_{i\in I}}(\Phi^{F})$ 
\end{center}
holds, where $\Phi^{F}$ is the winning set of the game $\game^{F}_{\rho}$.
\end{lem}

\begin{proof}
Consider the function $m_{X}$, as specified in Definition \ref{function_m_X}. Recall that, by definition, the equalities 
\begin{center}
$\mathbb{W}_{\beta} =   \bigcup_{\alpha<\beta}\mathbb{W}(\mathbb{W}_{\alpha})=m^{-1}_{\bigcup_{\alpha<\beta}\mathbb{W}_{\alpha}}(\Phi^{F})$
\end{center}
 hold. The proof is completed by showing that the following property holds:
\begin{equation}\label{equation_lemma_abcde}
\mathbb{P}_{M[M_{i}]_{i\in I}}( m_{\bigcup_{\alpha<\beta}\mathbb{W}_{\alpha}}^{-1}(X) ) = \mathbb{P}_{M[\gamma_{i}]_{i\in I}}(X),
\end{equation}
for every Borel measurable set $X\!\subseteq\! \bp^{F}$, where $\bp^{F}$ denotes the set of branching plays in $\game^{F}_{\rho}$. Indeed the desired result follows from Equation \ref{equation_lemma_abcde} by taking $X\!=\!\Phi^{F}$. Since probability measures in Polish spaces are regular, we can restrict $X$ to range over basic open sets. Equation \ref{equation_lemma_abcde} can then be proved, with routine techniques, by induction of the complexity of basic open sets $O_{T}$ (see Definition \ref{topology_bp}), i.e., on the size of the finite tree $T$.
\end{proof}

Note that, as observed earlier, the Markov branching play $M[\gamma_{i}]_{i\in I}$ considered in Lemma \ref{comparing_lemma}, might not be a real play in $\game^{F}_{\rho}$, i.e., one induced by a strategy profile. This is the case when $\gamma_{i}\!\neq\! \rho(X)(p_{i})$ for some $i\!\in\! I $.

\begin{figure}[t]
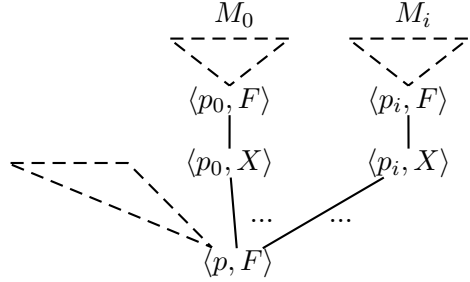

\pstree[ treemode=U,levelsep=8ex ]{\Tr{$\langle p, F\rangle$}}{
        \pstree[linestyle=none,arrows=-,levelsep=2ex]{\Tfan[fansize=10ex, linestyle=dashed]}{\TR{ $\ $}}
  
\pstree[levelsep=5ex]{\Tr{$\langle p_{0},X\rangle$} \trput{$...$}}{
	\pstree[levelsep=5ex]{\Tr{$\langle p_{0},F\rangle$}}{
		\pstree[linestyle=none,arrows=-,levelsep=2ex]{\Tfan[fansize=10ex, linestyle=dashed]}{\TR{ $M_{0} $}}
	
	}
}
\pstree[levelsep=5ex]{\Tr{$\langle p_{i},X\rangle$} \trput{$...$}}{
	\pstree[levelsep=5ex]{\Tr{$\langle p_{i},F\rangle$}}{
		\pstree[linestyle=none,arrows=-,levelsep=2ex]{\Tfan[fansize=10ex, linestyle=dashed]}{\TR{ $M_{i} $}}
	}
}
}
\caption{Markov Branching play $M{[}M_{i}{]_{i \in I}}$}
\label{figB1}
\end{figure}


\section{Proof of Equivalence of Game and Denotational Semantics}\label{proof_section}

This section is devoted to the proof of Theorem \ref{main_theorem}. The proof technique we adopt is based on the \emph{unfolding method} of \cite{FGK2010}.  The unfolding method can be roughly described as a technique for proving \emph{properties} of (some sort of) two-player parity games by induction on the number of priorities used in the game. Usually, the first step is to prove that the property under consideration holds for all  parity games with just one priority. Then the the general result for games with $n+1$ priorities follows by some argument making use of the inductive hypothesis. In our setting we are interested in two-player meta-parity games of the form $\game^{F}_{\rho}$, and the property we want to prove is that the lower and upper values of these games coincide with the denotational value of $F$ under the interpretation $\rho$.  We prove this by induction of the structure of $F$ rather than on the number of priorities used in the game $\game^{F}_{\rho}$. This allows a more transparent and arguably simpler proof. The structure of our proof closely resembles the one of \cite{MIO10}, where the equivalence of the game and denotational semantics for the logic pL$\mu$ is proved using the unfolding method. In the present context,  the proof is a significantly more technical undertaking due to the complexity of two-player meta-parity games: the $\mbox{\boldmath$\Delta$}^{1}_{2}$-complexity of the winning sets (Theorem \ref{delta_complexity}), their transfinite inductive characterization up to the first uncountable ordinal $\omega_{1}$ (Theorem \ref{iteration_lemma}) and the discontinuity in the free variables of the denotational semantics (see Theorem \ref{not_continuous_proposition} and Section \ref{robust_markov_plays_section}). In what follows, we shall focus primarily on the novel aspects of the proof, referring to \cite{MIO10} for a detailed analysis of the simpler cases that easily generalize to the present setting.\\

We prove, by induction on the structure of the formulas that, for every PLTS $\mathcal{L}\!=\!\langle P, \{ \freccia{a} \}_{a\in L}\rangle$, pL$\mu^{\odot}$ formula $F$ and $[0,1]$-interpretation $\rho$ of the variables, the following assertions hold:
\begin{equation}\label{main_equation_to_prove_1}
\sem{F}_{\rho}(p)=  \val_{\downarrow}( \game^{F}_{\rho})(\langle p, F\rangle) = \val_{\uparrow}( \game^{F}_{\rho})(\langle p, F\rangle),
\end{equation}
and
\begin{equation}\label{main_equation_to_prove_2} M  \textnormal{ is robust in the variable }X \textnormal{ (see Definition \ref{robust})} 
\end{equation}
for every free variable $X$ in $F$, and Markov branching play $M$ rooted at $\langle p, F\rangle$ in $\game^{F}_{\rho}$. \\

\textbf{Base case:} $\mathbf{F\!=\!X}$, for some variable $X\!\in\!\mathit{Var}$. \\
It follows immediately by application of Theorem \ref{fixed_point_proposition_1} that point \ref{main_equation_to_prove_1} holds. For what concerns point  \ref{main_equation_to_prove_2}, a Markov branching play $M$ rooted at $\langle p, X\rangle$ in $\game^{G}_{\rho}$ is of the following form:
\begin{center}
\pstree[treemode=U,levelsep=7ex]{\Tr{$\langle p_{0},X\rangle$} \trput{$...$}}{
	\Tr{$\top$}\tlput{$\lambda_{0}$}
	\Tr{$\bot$}\trput{$1-\lambda_{0}$}
}
\end{center}
where $\lambda_{0}\!=\!\rho(X)(p)$. Thus, $M$ has only one hole (i.e., $M=M[\lambda_{0}]$) and $\expected(M)\!=\!\lambda_{0}$. The Markov branching play $M[\gamma]$, can then be depicted as above, replacing $\lambda_{0}$ with $\gamma$, for every $\gamma\!\in\![0,1]$. Therefore, $\expected(M[\gamma])\!=\!\gamma$ and point \ref{main_equation_to_prove_2} is trivially satisfied as desired.\\

\textbf{Inductive cases} $\mathbf{F\!=\! G_{1} \star G_{2}}$ or $\mathbf{F\!=\! \circ G}$, for $\star\!\in\!\{\vee,\wedge,\cdot,\odot\}$ and $\circ\!\in\!\{\diam{a},\quadrato{a}\}$.\\
For all these cases, the proof of point \ref{main_equation_to_prove_1} follows easily by application of Theorem \ref{fixed_point_proposition_1}. The result can be proved following the same lines of the proof of (the corresponding inductive cases of) Theorem 5.1 in \cite{MIO10}, thus we omit the routine details.

We now prove that point \ref{main_equation_to_prove_2} holds for $F=G_{1}\cdot G_{2}$. The other cases can be proved in a similar way. Let us consider an arbitrary Markov branching play $M$ in $\game^{F}_{\rho}$ rooted at $\langle p, G_{1}\cdot G_{2}\rangle$. Then, $M$ can be depicted as follows, where $M_{1}$ and $M_{2}$ denote the two sub-Markov branching plays rooted at $\langle p, G_{1}\rangle$ and $\langle p, G_{2}\rangle$.
\begin{center}
\pstree[treemode=U,levelsep=7ex]{\Tr{$\langle p,G_{1}\cdot G_{2}\rangle$}}{
	\pstree[levelsep=5ex]{\Tr{$\langle p,G_{1}\rangle$}}{
		\pstree[linestyle=none,arrows=-,levelsep=2ex]{\Tfan[fansize=10ex, linestyle=dashed]}{\TR{ $M_{1} $}}
	}
	\pstree[levelsep=5ex]{\Tr{$\langle p,G_{2}\rangle$}}{
		\pstree[linestyle=none,arrows=-,levelsep=2ex]{\Tfan[fansize=10ex, linestyle=dashed]}{\TR{ $M_{2} $}}
	}
}
\end{center}
Note that, by definition of the game $\game^{F}_{\rho}$, the sub-Markov branching play $M_{i}$ is also a Markov branching play rooted at $\langle p, G_{i}\rangle$ in the game $\game^{G_{i}}_{\rho}$, for $i\!\in\!\{1,2\}$. Also note that $X$ is free in both $G_{1}$ and $G_{2}$, since it is free in $G$ by assumption. Let $M[x_{i}]_{i\in I}$ be the Markov branching pre-play obtained by pruning $M$ at the states of the form $\langle p_{i}, X\rangle$, as described in Section \ref{robust_markov_plays_section}. Let $\lambda_{i}$, for $i\!\in\! I$, be the probability labeling the edge connecting the $i$-th hole in $M$ (i.e., the state $\langle p_{i},X\rangle$), to the state $\top$. By Definition \ref{robust}, we need to prove that
\begin{enumerate}[(1)]
\item $\expected(M[\gamma_{i}]_{i\in I}) \geq \expected(M[\lambda_{i}]_{i\in I}) - \sum_{i \in I} \frac{\varepsilon}{2^{i+1}}$, and
\item $\expected(M[\delta_{i}]_{i\in I}) \leq \expected(M[\lambda_{i}]_{i\in I}) + \sum_{i \in I} \frac{\varepsilon}{2^{i+1}}$, 
\end{enumerate}
hold, for every sequence $\{\gamma_{i}\}_{i\in I}$ and $\{\delta_{i}\}_{i \in I}$ of reals in $[0,1]$ such that, for every $i\!\in\! I\!\subseteq\!\mathbb{N}$, the inequalities $\gamma_{i}\geq \lambda_{i}-\frac{\varepsilon}{\#(i)}$ and $\delta_{i}\leq \lambda_{i}+\frac{\varepsilon}{\#(i)}$ hold. We just show how to prove the first inequality because the second one can be proved in a similar way. 

Let $I_{1}$ and $I_{2}$ be the partition of the index set $I$ specified as follows. The index $i\!\in\! I$ of a  path $\vec{s}_{i}$ in  $M[x_{i}]_{i\in I}$  (connecting $\langle p, G_{1}\odot G_{2}\rangle$ with the hole $\langle p_{i},X\rangle$) is in $I_{1}$ if $\vec{s}_{i}$ lies in $M_{1}$, i.e., its second state is $\langle p, G_{1}\rangle$. Similarly, the index $i$ belongs to $I_{2}$ if $\vec{s}_{i}$ lies in $M_{2}$. Note that $I_{1}$ indexes all the holes of the sub-Markov branching play $M_{1}$ and, similarly, $I_{2}$ indexes the holes of $M_{2}$. Thus, let us denote with $M_{1}[x_{i}]_{i\in I_{1}}$ the Markov branching pre-play associated with $M_{1}$, and similarly for $M_{2}[x_{j}]_{j\in I_{2}}$ and $M_{2}$. Then, by the inductive hypothesis on the complexity of $G_{1}$ and $G_{2}$, we know that $M_{1}$ and $M_{2}$ are robust in $X$, i.e., the following assertions, with respect to inequality ($1$) above,
\begin{enumerate}[a)]
\item $\expected(M_{1}[\gamma_{i}]_{i\in I_{1}}) \geq \expected(M_{1}[\lambda_{i}]_{i\in I_{1}}) - \sum_{i \in I_{1}} \frac{\varepsilon}{2^{i+1}}$, and
\item $\expected(M_{2}[\gamma_{j}]_{j\in I_{2}}) \geq \expected(M_{2}[\lambda_{j}]_{j\in I_{2}}) - \sum_{j \in I_{2}} \frac{\varepsilon}{2^{j+1}}$, 
\end{enumerate}
hold. By applying the product measure technique adopted in the proof of Theorem \ref{fixed_point_proposition_1}, it is easy to verify that the  equalities $\expected(M[\lambda_{i}]_{i\in I})\!=\! \expected(M_{1}[\lambda_{i}]_{i\in I_{1}}) \cdot \expected(M_{2}[\lambda_{j}]_{j\in I_{2}})$ and $\expected(M[\gamma_{i}]_{i\in I})\!=\! \expected(M_{1}[\gamma_{i}]_{i\in I_{1}}) \cdot \expected(M_{2}[\gamma_{j}]_{j\in I_{2}})$ hold. The desired equation $(1)$ above can then be derived as follows:
\begin{center}
\begin{tabular}{l l l }
$\expected(M[\gamma_{i}]_{i\in I})$ & $=$ & $\expected(M_{1}[\gamma_{i}]_{i\in I_{1}}) \cdot \expected(M_{2}[\gamma_{j}]_{j\in I_{2}})$ \\
$$ & $\geq_{b)}$ & $\expected(M_{1}[\gamma_{i}]_{i\in I_{1}}) \cdot \big(  \expected(M_{2}[\lambda_{j}]_{j\in I_{2}}) - \sum_{j \in I_{2}} \frac{\varepsilon}{2^{j+1}}     \big) $\\
$$ & $\geq_{*}$ & $ \big( \expected(M_{1}[\gamma_{i}]_{i\in I_{1}}) \cdot \expected(M_{2}[\lambda_{j}]_{j\in I_{2}}) \big)  -  \sum_{j \in I_{2}} \frac{\varepsilon}{2^{j+1}}    $\\
$$ & $\geq_{a)}$ & $ \Big( \big(\expected(M_{1}[\lambda_{i}]_{i\in I_{1}}) - \sum_{i \in I_{1}} \frac{\varepsilon}{2^{i+1}} \big) \cdot \expected(M_{2}[\lambda_{j}]_{j\in I_{2}})  \Big)  - \sum_{j \in I_{2}} \frac{\varepsilon}{2^{j+1}}  $\\
$$ & $\geq_{*}$ & $\big( \expected(M_{1}[\lambda_{i}]_{i\in I_{1}}) \cdot \expected(M_{2}[\lambda_{j}]_{j\in I_{2}}) \big)  - \sum_{i \in I_{1}} \frac{\varepsilon}{2^{i+1}} - \sum_{j \in I_{2}} \frac{\varepsilon}{2^{j+1}}  $\\
$$ & $\geq_{I=I_{1}\uplus I_{2}}$ & $\big( \expected(M_{1}[\lambda_{i}]_{i\in I_{1}}) \cdot \expected(M_{2}[\lambda_{j}]_{j\in I_{2}}) \big)  - \sum_{i \in I} \frac{\varepsilon}{2^{i+1}} $\\
$$ & $=$ & $\expected(M[\lambda_{i}]_{i\in I})  - \sum_{i \in I} \frac{\varepsilon}{2^{i+1}} $\\
\end{tabular}
\end{center}
where the steps labeled with ($*$) are valid because all terms have values in the interval $[0,1]$.\\


\textbf{Inductive case} $\mathbf{F\!=\! \mu X.G}$. \\
\begin{itemize}
\item We first prove that point \ref{main_equation_to_prove_1} holds. 
\end{itemize}
For every  state $p$ and every interpretation $\rho$ we have, by definition of the denotational semantics, that the following equality holds:
\begin{center}
$\sem{\mu X.G}_{\rho}(p)\bydef \lfp \Big( \lambda f\!\in\![0,1]^{P}. \big( \sem{G}_{\rho[f/X]}\big)  \Big)(p)$.
\end{center}
By the Knaster--Tarski theorem, the previous equation can be rewritten as:
\begin{equation}\label{fix_step_aux_1}
\sem{\mu X.G}_{\rho}(p)=\bigsqcup_{\alpha} \sem{G}^{\alpha}_{\rho},
\end{equation}
where $\alpha$ ranges over the ordinals, and $\sem{G}^{\alpha}_{\rho}$ is defined as $\bigsqcup_{\beta<\alpha}\sem{G}_{\rho[\sem{G}^{\beta}_{\rho}/X]}$. Let us denote with $\gamma$ the least ordinal such that $\sem{G}^{\gamma}_{\rho}=\sem{\mu X.G}_{\rho}$, and with $\rho_{\gamma}\!\in\![0,1]^{P}$ the interpretation $\rho[\sem{G}^{\gamma}_{\rho}/X]$. Thus,  the following equation holds: 
\begin{equation}
\sem{G}_{\rho_{\gamma}}=\sem{\mu X.G}_{\rho}.
\end{equation}

Let us now turn our attention to the $2\frac{1}{2}$-player meta-parity game $\game^{\mu X.G}_{\rho}$. Our goal is to prove that point \ref{main_equation_to_prove_1}  holds, i.e., that the following equalities
\begin{equation}\label{fix_step_aux_2}
\sem{\mu X.G}_{\rho}(p) = \val_{\downarrow}\big(\game^{\mu X.G}_{\rho}\big)(\langle p, \mu X.G\rangle)= \val_{\uparrow}\big(\game^{\mu X.G}_{\rho}\big)(\langle p, \mu X.G\rangle)
\end{equation}
hold, for every $p\!\in\!P $. As a first observation, note that the state $\langle p, \mu X.G\rangle$ is not reachable by any other game state, and that it has the state $\langle p, G\rangle$ as its only successor state. It then follows by application of Proposition \ref{fixed_point_proposition_1} that, in order to prove the desired result (\ref{fix_step_aux_2}), we just have to show that the equalities 
\begin{equation}\label{fix_step_aux_3}
\sem{G}_{\rho_{\gamma}}(p) = \val_{\downarrow}\big(\game^{\mu X.G}_{\rho}\big)(\langle p,G\rangle)= \val_{\uparrow}\big(\game^{\mu X.G}_{\rho}\big)(\langle p, G\rangle)
\end{equation}
hold. In order to improve readability, we shall denote with $\gsem{\mu X.G}^{\star}_{\rho}\!:\!P\!\rightarrow\![0,1]$ the function defined as $\lambda p\!\in\!P.\Big( \val_{\star}\big(\game^{\mu X.G}_{\rho}\big)(\langle p,G\rangle) \Big)$, for $\star\!\in\!\{\downarrow,\uparrow\}$. Thus, Equation \ref{fix_step_aux_3} can be rewritten as follows:
\begin{equation}\label{fix_step_aux_3_prime}
\sem{G}_{\rho_{\gamma}}(p) = \gsem{\mu X.G}^{\downarrow}_{\rho}(p)=\gsem{\mu X.G}^{\uparrow}_{\rho}(p).
\end{equation}
Note that the analogous functions $\gsem{G}^{\star}_{\rho[f/X]}\!:\!P\!\rightarrow\![0,1]$ specified, for $\star\!\in\!\{\downarrow,\uparrow\}$, as \\ $\lambda p\!\in\!P.\Big( Val_{\star}\big(\game^{G}_{\rho[f/X]}\big)(\langle p,G\rangle) \Big)$, satisfy the following equation:
\begin{equation}\label{fix_step_aux_3_second}
\sem{G}_{\rho[f/X]}= \gsem{G}^{\downarrow}_{\rho[f/X]} = \gsem{G}^{\uparrow}_{\rho[f/X]} 
\end{equation}
for all $f\!\in\![0,1]^{P}$, by induction hypothesis (\ref{main_equation_to_prove_1}) on $G$.

We prove Equation \ref{fix_step_aux_3_prime} by exploiting the similarities between the game $\game^{\mu X.G}_{\rho}$ and the game $ \game^{G}_{\rho[f/X]}$, already discussed in Section \ref{technical_section}. Our first observation is the following:
\begin{equation}\label{fix_step_aux_6}
\gsem{\mu X.G}^{\star}_{\rho} =  \gsem{G}^{\star}_{\rho[\gsem{\mu X.G}^{\star}_{\rho}/X]} \stackrel{\textnormal{Eq. \ref{fix_step_aux_3_second}}}{=} \sem{G}_{\rho[\gsem{\mu X.G}^{\star}_{\rho}/X]}
\end{equation}
for $\star\!\in\!\{\downarrow,\uparrow\}$. Indeed when a state of the form $\langle p, X\rangle$ is reached in the game $\game^{\mu X.G}_{\rho}$ the play continues from the state $\langle p, G\rangle$ and ends in a victory for Player $1$ with (lower or upper) value $\gsem{\mu X.G}^{\star}_{\rho}(p)$, and, similarly, when the state $\langle p, X\rangle$ is reached in $\game^{G}_{\rho[\gsem{\mu X.G}^{\star}_{\rho}/X]}$,  the play immediately terminates in favor of Player $1$ with probability $\gsem{\mu X.G}^{\star}_{\rho}(p)$.

By application of Equation \ref{fix_step_aux_3_second}, this implies that both $\gsem{\mu X.G}^{\uparrow}_{\rho}$ and $\gsem{\mu X.G}^{\downarrow}_{\rho}$ are  fixed points of the functional $\lambda f\!\in\![0,1]^{P}.( \sem{G}_{\rho[f/X]})$. Note that, for all $p\!\in\!P$, the inequality $\gsem{\mu X.G}^{\downarrow}_{\rho}(p) \leq \gsem{\mu X.G}^{\uparrow}_{\rho}(p)$ trivially holds. Moreover the inequality $\sem{\mu X.G}_{\rho}(p) \leq \gsem{\mu X.G}^{\uparrow}_{\rho}(p)$ holds, for all $p\!\in\!P$, because $\sem{\mu X.G}_{\rho}$ (or, equivalently, $\sem{G}_{\rho_{\gamma}}$) is the least fixed point of $\lambda f\!\in\![0,1]^{P}.( \sem{G}_{\rho[f/X]})$.

We shall prove the desired result (Equation \ref{fix_step_aux_3_prime}) by showing that, for all $p\!\in\!P$, the inequality
\begin{equation}
\gsem{\mu X.G}^{\uparrow}_{\rho}(p)\bydef Val^{\uparrow}\big( \game^{\mu X.G}_{\rho}\big)(\langle p, G\rangle) \leq   \sem{G}_{\rho_{\gamma}}(p)
\end{equation}
holds. We do this by constructing, for every $\varepsilon\!>\!0$ a strategy $\sigma^{\varepsilon}_{2}$ for Player $2$ in the game $\game^{\mu X.G}_{\rho}$, satisfying the following inequality:
 \begin{equation}\label{strategy_eq_1}
  \bigsqcup_{\sigma_{1}}\expected(M^{\langle p, G\rangle}_{\sigma_{1},\sigma^{\varepsilon}_{2}}) \leq   \sem{G}_{\rho_{\gamma}}(p) + \varepsilon.
  \end{equation}

The strategy $\sigma^{\varepsilon}_{2}$ is constructed using the collection of $\delta$-optimal strategies, for $\delta\!>\!0$, in the game $\game^{G}_{\rho_{\gamma}}$, i.e., strategies $\tau^{\delta}_{2}$ such that the following equality holds:
\begin{equation}\label{delta_optima_strategies}
\bigsqcup_{\tau_{1}}\expected(M^{\langle p, G\rangle}_{\tau_{1},\tau^{\delta}_{2}}) \leq   \val\big(\game^{G}_{\rho_{\gamma}}\big)(\langle p, G\rangle) + \delta \stackrel{\textnormal{Eq. \ref{fix_step_aux_3_second}}}{=}  \sem{G}_{\rho_{\gamma}}(p) + \delta.
\end{equation} 
Let $e\!:\!\mathcal{P}^{<\omega}_{\mu X.G}\!\rightarrow\!\mathbb{N}$ be a numbering (i.e., an injective map into the natural numbers) of the finite paths in the game $\game^{\mu X.G}_{\rho}$. The strategy $\sigma^{\varepsilon}_{2}$ is defined, for every $\varepsilon \!>\!0$,  as follows:
\begin{center}
$\sigma^{\varepsilon}_{2}(\vec{s}) =   \left\{     \begin{array}{l  l}
				\tau^{\frac{\varepsilon}{2}}_{2}(\vec{s}) & $if $\vec{s} $ does not contain states of the form $\langle p, X\rangle $, for $ p\!\in\!P \\
					\\
					\sigma^{ \frac{\varepsilon}{2}\cdot \frac{1}{\#(e(\vec{s}_{j}))}  }_{2}(\vec{t})  &   $ if $\vec{s}\!=\!\vec{s}_{j}.\vec{t}$ with $\last(\vec{s}_{j})=\langle p, X\rangle $ for some $p\!\in\!P \\				 \end{array}      \right.$
\end{center}
for every finite path $\vec{s}$ whose last state belong to Player $2$ (i.e., such that $\last(\vec{s})\!\in\! S_{2}$), where the function $\#\!:\!\mathbb{N}\!\rightarrow\!\mathbb{N}$ is specified as in Definition \ref{approx_function}. The strategy $\sigma^{\varepsilon}_{2}$ can be informally described as follows. When the game begins $\sigma^{\varepsilon}_{2}$ behaves as the strategy $\tau^{\frac{\varepsilon}{2}}_{2}$ until a state of the form $\langle p, X\rangle$ is reached. This is a good definition because  plays in the two games $\game^{\mu X.G}_{\rho}$ and $\game^{G}_{\rho_{\gamma}}$ are identical until states of this form are reached. If a state of the form $\langle p, X\rangle$ is eventually reached following some path $\vec{s}_{j}$, then Player $2$ \emph{improves} their strategy, and plays the rest of the game (starting at the unique successor state $\langle p, G\rangle$ of $\langle p, X\rangle$) in accordance with the better strategy $\sigma^{ \frac{\varepsilon}{2}\cdot \frac{1}{\#(e(\vec{s}_j))}  }_{2}$. Note how the choice of the new strategy to be followed crucially depends on the path $\vec{s}_{j}$.

We are now going to show that, for every $\varepsilon \!>\!0$,  the strategy $\sigma^{\varepsilon}_{2}$ satisfies the  desired inequality \ref{strategy_eq_1}. We need to show that, for every strategy $\sigma_{1}$ for Player $1$, the inequality 
\begin{equation}
\expected(M^{\langle p, G\rangle}_{\sigma_{1},\sigma^{\varepsilon}_{2}}) \leq   \sem{G}_{\rho_{\gamma}}(p) + \varepsilon.
\end{equation} 
holds. Recall that, by definition, $\expected(M^{s}_{\sigma_{1},\sigma^{\varepsilon}_{2}})\!=\!\mathbb{P}^{s}_{\sigma_{1},\sigma^{\varepsilon}_{2}}(\Phi)$, where $\mathbb{P}^{\langle p, G\rangle}_{\sigma_{1},\sigma^{\varepsilon}_{2}}$ denotes the probability measure over branching plays in $\game^{\mu X.G}_{\rho}$ induced by the Markov branching play $M^{s}_{\sigma_{1},\sigma^{\varepsilon}_{2}}$, and $\Phi$ denotes  the set of winning branching plays for Player $1$ in $\game^{\mu X.G}_{\rho}$. By Theorem \ref{inductive_ch} and Lemma \ref{iteration_lemma}, we know that $\Phi\!=\! \bigcup_{\alpha<\omega_{1}} \mathbb{W}_{\alpha}$. Thus, the previous equation can be rewritten as follows:
\begin{equation}
\mathbb{P}^{\langle p, G\rangle}_{\sigma_{1},\sigma^{\varepsilon}_{2}}\big( \bigcup_{\alpha<\omega_{1}} \mathbb{W}_{\alpha} \big)  \leq   \sem{G}_{\rho_{\gamma}}(p) + \varepsilon.
\end{equation} 
By application of Theorem \ref{consequences_martin}, under the set-theoretic assumption $\textnormal{MA}_{\aleph_{1}}$ we can further rewrite the previous equation as follows:
\begin{equation}\label{step_martin_setup}
\bigsqcup_{\alpha<\omega_{1}}\big( \mathbb{P}^{\langle p, G\rangle}_{\sigma_{1},\sigma^{\varepsilon}_{2}}( \mathbb{W}_{\alpha}) \big)  \leq   \sem{G}_{\rho_{\gamma}}(p) + \varepsilon.
\end{equation} 
This step allow us to set up a proof by transfinite induction for the desired Equality \ref{strategy_eq_1}. We shall now prove that for every countable ordinal $\beta<\omega_{1}$ and every $\varepsilon >0$, the inequality 
\begin{equation}\label{inductive_hp_equation}
 \mathbb{P}^{\langle p, G\rangle}_{\sigma_{1},\sigma^{\varepsilon}_{2}}( \mathbb{W}_{\beta}) \leq   \sem{G}_{\rho_{\gamma}}(p) + \varepsilon.
\end{equation}
holds. Assume Equation \ref{inductive_hp_equation} holds for every ordinal $\alpha<\beta$. The Markov branching play $M^{\langle p,G\rangle}_{\sigma_{1},\sigma^{\varepsilon}_{2}}$ can be depicted, following the notation introduced in sections \ref{technical_section} and \ref{robust_markov_plays_section}, as in Figure \ref{fig_proof_A1}.
\begin{figure}[t]
\pstree[ treemode=U,levelsep=8ex ]{\Tr{$\langle p, G\rangle$}}{
        \pstree[linestyle=none,arrows=-,levelsep=2ex]{\Tfan[fansize=10ex, linestyle=dashed]}{\TR{ $\ $}}
  
\pstree[levelsep=5ex]{\Tr{$\langle p_{0},X\rangle$} \trput{$...$}}{
	\pstree[levelsep=5ex]{\Tr{$\langle p_{0},G\rangle$}}{
		\pstree[linestyle=none,arrows=-,levelsep=2ex]{\Tfan[fansize=10ex, linestyle=dashed]}{\TR{ $M_{0} $}}
	
	}
}
\pstree[levelsep=5ex]{\Tr{$\langle p_{i},X\rangle$} \trput{$...$}}{
	\pstree[levelsep=5ex]{\Tr{$\langle p_{i},G\rangle$}}{
		\pstree[linestyle=none,arrows=-,levelsep=2ex]{\Tfan[fansize=10ex, linestyle=dashed]}{\TR{ $M_{i} $}}
	}
}
}
\caption{Markov Branching play $M{[}M_{i}{]_{i \in I}}$}
\label{fig_proof_A1}
\end{figure}
By definition of $\sigma^{\varepsilon}_{2}$, the sub-Markov branching plays $M_{i}$, for $i\!\in\!I$, are Markov branching plays played in accordance with the strategy $\sigma^{\varepsilon_{i}}_{2}$, with $\varepsilon_{i}\!=\! \frac{\varepsilon}{2}\cdot \frac{1}{\#(e(\vec{s}_{i}))}$. Let us denote with $\delta^{\alpha}_{i}$ the value $\mathbb{P}_{M_{i}}(\mathbb{W}_{\alpha})$, where $\mathbb{P}_{M_{i}}$ denotes the probability measure associated with $M_{i}$, and $\alpha <\beta$. Moreover, let us denote with $\delta_{i}$, for $i\!\in\!I$ the value $\bigsqcup_{\alpha<\beta}\delta^{\alpha}_{i}$.
Then, by induction hypothesis on the ordinals, we know that, for every $\alpha\! <\!\beta$, the inequality  
\begin{equation}
\delta^{\alpha}_{i} \leq  \sem{G}_{\rho_{\gamma}}(p) + \frac{\varepsilon}{2}\cdot \frac{1}{\#(e(\vec{s}_{i}))}
\end{equation} 
holds, for every $i\!\in\! I$. By application of Lemma \ref{comparing_lemma}, we know that the following equality holds:
\begin{equation}
 \mathbb{P}^{\langle p,G\rangle}_{\sigma_{1},\sigma^{\varepsilon}_{2}}(\mathbb{W}_{\beta}) = \expected( M[\delta_{i}]_{i\in I}])
\end{equation}
where $M[\delta_{i}]_{i\in I}$ denotes the Markov branching play in the game $\game^{G}_{\rho^{\gamma}}$, obtained from $M^{\langle p, G\rangle}_{\sigma_{1},\sigma^{\varepsilon}_{2}}$ as specified in Section \ref{robust_markov_plays_section}, which can be depicted as in Figure \ref{fig_proof_B1-A}.
\begin{figure}
	\pstree[ treemode=U,levelsep=8ex ]{\Tr{$\langle p, G\rangle$}}{
        \pstree[linestyle=none,arrows=-,levelsep=2ex]{\Tfan[fansize=10ex, linestyle=dashed]}{\TR{ $\ $}}
  
\pstree[levelsep=7ex]{\Tr{$\langle p_{0},X\rangle$} \trput{$...$}}{
	\Tr{$\top$}\tlput{$\delta_{0}$}
	\Tr{$\bot$}\trput{$1-\delta_{0}$}
}
\pstree[levelsep=7ex, linestyle=none]{\Tr{$\dots $}}{
	\Tr{$\ $}
	\Tr{$\ $}
}
\pstree[levelsep=7ex]{\Tr{$\langle p_{i},X\rangle$} \trput{$...$}}{
	\Tr{$\top$}\tlput{$\delta_{i}$}
	\Tr{$\bot$}\trput{$1-\delta_{i}$}
}
}	
\caption{Markov branching play $M[\delta_{i}]_{i\in I}$}
\label{fig_proof_B1-A}
\end{figure}
Let us denote with $\lambda_{i}$, for $i\!\in\!I$, the value $\sem{G}_{\rho_{\gamma}}(p_{i})$. Observe that the Markov branching play $M[\lambda_{i}]_{i\in I}$ (depicted by replacing $\delta_{i}$ with $\lambda_{i}$ in Figure \ref{fig_proof_B1-A}) is a Markov branching play in the game $\game^{G}_{\rho^{\gamma}}$. Moreover, by definition of $\sigma^{\varepsilon}_{2}$, the Markov branching play $M[\lambda_{i}]_{i\in I}$ is played by Player $2$ in accordance with the $\frac{\varepsilon}{2}$-optimal strategy $\tau^{\frac{\varepsilon}{2}}_{2}$. It then follows by Equation \ref{delta_optima_strategies} that
\begin{equation}\label{quasi_finito_1}
\expected(M[\lambda_{i}]_{i\in I}) \leq \sem{G}_{\rho_{\gamma}}(p) + \frac{\varepsilon}{2}
\end{equation}
Recall that, by induction hypothesis (\ref{main_equation_to_prove_2}) on $G$, the Markov branching play $M[\lambda_{i}]_{i\in I}$ is robust is the free (in $G$) variable $X$. Thus the following inequality holds: 
\begin{equation}\label{quasi_finito_2}
\expected(M[\delta_{i}]_{i\in I})  \leq \expected(M[\lambda_{i}]_{i\in I}) + \sum_{i\in I}\big(\frac{\varepsilon}{2}\cdot \frac{1}{\#(e(\vec{s}_{i}))}\big).
\end{equation}
Since the numbering $e$ is injective, it follows from equation \ref{quasi_finito_1} and \ref{quasi_finito_2} that $\expected(M[\delta_{i}]_{i\in I})  \leq \sem{G}_{\rho_{\gamma}}(p) -  \frac{\varepsilon}{2} -\frac{\varepsilon}{2}$. Thus, Equation \ref{inductive_hp_equation} holds as desired.\\


\begin{figure}
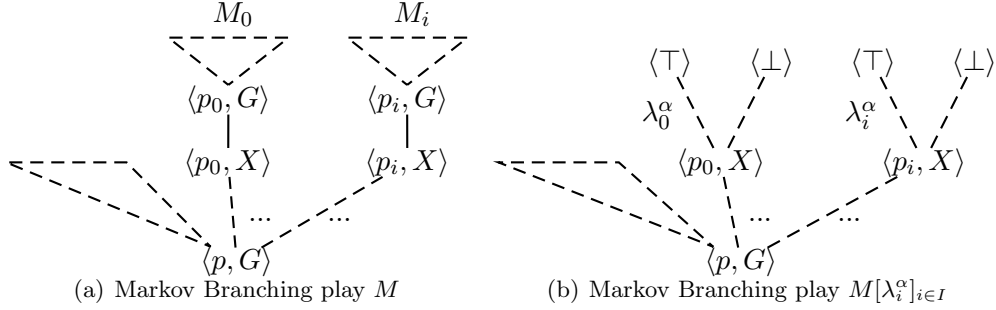

\centering
\subfigure[Markov Branching play $M$]{\label{fig20}
	\pstree[ treemode=U,levelsep=8ex,linestyle=dashed]{\Tr{$\langle p, G\rangle$}}{
        \pstree[linestyle=none,arrows=-,levelsep=2ex]{\Tfan[linestyle=dashed,fansize=10ex]}{\TR{ $\ $}}
  
\pstree[levelsep=5ex,linestyle=solid]{\Tr{$\langle p_{0},X\rangle$} \trput{$...$}}{
	\pstree[levelsep=5ex]{\Tr{$\langle p_{0},G\rangle$}}{
		\pstree[linestyle=none,arrows=-,levelsep=2ex]{\Tfan[linestyle=dashed,fansize=10ex]}{\TR{ $M_{0} $}}
	
	}
}
\pstree[levelsep=5ex,linestyle=solid]{\Tr{$\langle p_{i},X\rangle$} \trput{$...$}}{
	\pstree[levelsep=5ex]{\Tr{$\langle p_{i},G\rangle$}}{
		\pstree[linestyle=none,arrows=-,levelsep=2ex]{\Tfan[linestyle=dashed,fansize=10ex]}{\TR{ $M_{i} $}}
	}
}
}
}
\subfigure[Markov Branching play $M{[}\lambda^{\alpha}_{i}{]_{i \in I}}$]{\label{fig21}
	\pstree[ treemode=U,levelsep=8ex,linestyle=dashed ]{\Tr{$\langle p, G\rangle$}}{
        \pstree[linestyle=none,arrows=-,levelsep=2ex]{\Tfan[linestyle=dashed,fansize=10ex]}{\TR{ $\ $}}

	\pstree[levelsep=8ex]{\Tr{$\langle p_{0},X\rangle$} \trput{$...$}}{\psset{linestyle=dashed}  \Tr{$\langle\top\rangle $}\tlput{$\lambda^{\alpha}_{0}$} \Tr{$  \langle\bot\rangle  $}\trput{$$}} 
	\pstree[levelsep=8ex]{\Tr{$\langle p_{i},X\rangle$} \trput{$...$}}{\psset{linestyle=dashed}  \Tr{$\langle\top\rangle$}\tlput{$\lambda^{\alpha}_{i}$} \Tr{$\langle\bot\rangle $}\trput{$$}} 
	}
}

\caption{Markov branching plays $M\!=\!M[M_{i}]_{i\in I}$ in $\game^{\mu X.G}_{\rho}$ and $M{[}\lambda^{\alpha}_{i}{]_{i \in I}}$ in $\game^{G}_{\rho}$}
\end{figure}

\begin{itemize}
\item We now prove that point \ref{main_equation_to_prove_2} holds.
\end{itemize}
We just discuss the main ideas of the proof as the necessary techniques have been already introduced in the proof of point \ref{main_equation_to_prove_1} above.

We need to show that every Markov branching play $M$ in $\game^{\mu X.G}_{\rho}$, rooted at $\langle p, \mu X.G\rangle$ is robust in every free variable $Y$. Clearly $Y\!\neq\! X$, since $X$ is bound in $\mu X.F$. As already observed earlier, we can just consider Markov branching plays $M$ rooted at $\langle p, G\rangle$, since the state $\langle p, \mu X.G\rangle$ has $\langle p,G\rangle$ has its only successor state, and it is not reachable by any other states. Let $J\!\subseteq\!\mathbb{N}$ be the set indexing the paths from the root $\langle p, G\rangle$ of $M$ to states of the form $\langle p_{j},Y\rangle$. Let us denote with $M[y_{j}]_{j\in J}$ the associated Markov branching pre-play. Let $\lambda_{j}\!\in\![0,1]$, for $j\!\in\! J$, be the value labeling the edge connecting the $j$-th hole in $M$ with the state $\top$, i.e., $\lambda_{j}=\rho(Y)(p_{j})$.  Thus $M=M[\lambda_{j}]_{j\in J}$. We need to prove that, for every $\varepsilon >0$,  the inequalities
\begin{enumerate}[(1)]
\item $\expected(M[\gamma_{j}]_{j\in J}) \geq \expected(M[\lambda_{j}]_{j\in J}) - \sum_{j \in J} \frac{\varepsilon}{2^{j+1}}$, and
\item $\expected(M[\delta_{j}]_{j\in J}) \leq \expected(M[\lambda_{j}]_{j\in J}) + \sum_{j \in J} \frac{\varepsilon}{2^{j+1}}$, 
\end{enumerate}
hold for every sequences  $\{\gamma_{j}\}_{j\in J}$ and $\{\delta_{j}\}_{j \in J}$ of reals in $[0,1]$ such that, for every $j\!\in\! J\!\subseteq\!\mathbb{N}$, the inequalities $\gamma_{j}\geq \lambda_{j}-\frac{\varepsilon}{\#(j)}$ and $\delta_{j}\leq \lambda_{j}+\frac{\varepsilon}{\#(j)}$ hold. We just consider the first inequality, as the second one can be proved in a similar way. 

Recall that, by application of Theorem \ref{consequences_martin}, under the set-theoretic assumption $\textnormal{MA}_{\aleph_{1}}$,  the equality $\expected(M[\gamma_{j}]_{j\in J})\!\bydef\! \mathbb{P}_{M[\gamma_{j}]_{j\in J}}(\Phi)  \!=\!\bigsqcup_{\alpha<\omega_{1}}\mathbb{P}_{M[\gamma_{j}]_{j\in J}}(\mathbb{W}_{\alpha})$ holds. Similarly for $\expected(M[\lambda_{j}]_{j\in J})$. We shall then prove ($1$) above, by proving  the following more general property: for all Markov branching plays $M=M[\lambda_{j}]_{j\in J}$ in $\game^{\mu X.G}_{\rho}$ rooted at $\langle p, G\rangle$ and sequence  $\{\gamma_{j}\}_{j\in J}$ as described above, the following equality holds:
\begin{equation}\label{transfinite_hp_rubust}
\mathbb{P}_{M[\gamma_{j}]_{j \in J}}(\mathbb{W}_{\alpha}) \geq \mathbb{P}_{M[\lambda_{j}]_{j\in J}}(\mathbb{W}_{\alpha})- \sum_{j \in J} \frac{\varepsilon}{2^{j+1}}
\end{equation}
This is proven, again, by transfinite induction on the ordinals. Suppose the property holds for all $\alpha <\beta$. The main idea to prove the inductive step is to reduce the problem on $M$ (which can be depicted as in Figure \ref{fig20} by exposing the collection of paths reaching states of the form $\langle p_{i},X\rangle$, for an index set $I$) to that of $M[\lambda^{\alpha}_{i}]_{i\in I}$ (where $\lambda^{\alpha}_{i}\!=\! \mathbb{P}_{M_{i}}(\bigcup_{\alpha<\beta}\mathbb{W}_{\alpha})$), the Markov branching play in $\game^{G}_{\rho}$ constructed as described in Lemma \ref{comparing_lemma} and depicted in Figure \ref{fig21}. The result then follows by applications of Lemma \ref{comparing_lemma} (relating the values $\mathbb{P}_{M}(\mathbb{W}_{\beta})$ and of $\expected(\mathbb{P}_{M[\lambda^{\alpha}_{i}]})$), by induction hypothesis on $\alpha$ (\ref{transfinite_hp_rubust}) and by induction hypothesis  (\ref{main_equation_to_prove_2}) on the robustness of  the Markov branching plays $M[\lambda^{\alpha}_{i}]_{i\in I}$ in $\game^{G}_{\rho}$. We omit the technical details. A detailed proof of (a generalized version of) this result can be found in \cite[\S 4.3]{MioThesis}.\\

\textbf{Inductive case} $\mathbf{G\!=\! \nu X.G}$. \\
Similar to the previous one and based on the properties summarized in Proposition \ref{dual_results}.


\subsection{Remarks about the use of $\textnormal{MA}_{\aleph_{1}}$}\label{remarks_martin}
We conclude this section by highlighting the two critical steps in our proof where the set-theoretic assumption $\textnormal{MA}_{\aleph_{1}}$ is used. 

As already discussed in Section \ref{tree_games_section}, we make use of Martin's Axiom at $\aleph_{1}$ to ensure that the $\mbox{\boldmath$\Delta$}^{1}_{2}$ winning set of a pL$\mu^{\odot}$ game is universally measurable. The universal measurability of $\mbox{\boldmath$\Delta$}^{1}_{2}$ sets can, however, be proved in other extensions of $\textnormal{ZFC}$. For example, determinacy-based axioms such as \emph{Analytic Determinacy} \cite{Jech}, suffice (see, e.g., \cite[\S 36.E]{Kechris}). 

The second use we make of the axiom $\textnormal{MA}_{\aleph_{1}}$ is in the derivation of Equation \ref{step_martin_setup} in the proof above. This is a fundamental step required to set up a proof by transfinite induction on the countable ordinals. As stated in Theorem \ref{consequences_martin}, one of the consequences of Martin's Axiom at $\aleph_{1}$ is that probability measures on Polish spaces are $\omega_{1}$-continuous. Such a property, clearly implying the negation of the Continuum hypothesis, does not follow from determinacy-based axioms mentioned above. As mentioned after Lemma \ref{iteration_lemma}, our inductive characterization of the winning set of pL$\mu^{\odot}$ games, can be shown to require, in general, $\omega_{1}$-iterations (i.e., approximants) to converge. Thus some form of $\omega_{1}$-continuity seems to be required by the proof technique adopted in this Section.

\section{Conclusions and Future Work}\label{conclusion_section}

One of the primary interests in a game semantics for pL$\mu^{\odot}$, and more generally for all logics having a $[0, 1]$-valued semantics with an intended probabilistic reading, is to offer an accessible and clear interpretation for the property described by a formula. We suggest that our game semantics, built on top of the elementary idea of concurrent execution of independent sub-istances of the game, succeeds in this task. The logic pL$\mu^{\odot}$ is expressive enough to encode the qualitative threshold modalities $\mathbb{P}_{>0}$ and $\mathbb{P}_{=1}$ which, as discussed at the end of  Section \ref{model_checking_games}, are interpreted within the game semantics in a straightforward way. 

Despite the naturalness of the definition, our proof of equivalence is a technically undertaking. Our result, based on a transfinite characterization of pL$\mu^{\odot}$ winning sets, is carried out in $\textnormal{ZFC}+\textnormal{MA}_{\aleph_{1}}$ set-theory to deal with the measure theoretic complications associated with the complexity of winning sets. We are not aware of any other result in theoretical computer science whose proof is (or at least was originally) carried out in proper extensions of $\textnormal{ZFC}$ set theory. Thus, our proof of determinacy, which is generalized in \cite{MioThesis} from the pL$\mu^{\odot}$ games considered in this paper to arbitrary two-player stochastic meta-parity games,  is perhaps noteworthy as being a first example of this kind of result.

Although the logic pL$\mu^{\odot}$ subsumes the qualitative fragment of the logic PCTL of \cite{BA1995}, it does not seem possible to encode the full logic PCTL within pL$\mu^{\odot}$. In \cite{MioThesis} an extension of pL$\mu^{\odot}$,  obtained by adding to the syntax of the logic yet another pair of De Morgan dual connectives,  is considered. Interestingly, this extension, which is capable of encoding the full logic PCTL, can be given an appropriate game semantics in terms of two-player stochastic meta-parity games. We refer to \cite{MioThesis} for a proof of this fact, which serves as a demonstration of the expressive power of the new class of two-player stochastic meta-parity games introduced in this paper.

Our work leaves open several directions for future research. From a theoretical point of view, it would be interesting to remove the dependencies on the set-theoretic assumption $\textnormal{MA}_{\aleph_{1}}$ from our proof. This, in light of the remarks of Section \ref{remarks_martin}, looks as a challenging task. In another direction, it would be interesting to investigate the theory of two-player stochastic meta-games. Preliminary results, such as the fact that Blackwell games can be encoded as tree games and that the open problem of \emph{qualitative determinacy} \cite{vaclav2011a}  for stochastic games can be formulated as a determinacy problem for tree games, are obtained in the author's PhD thesis \cite{MioThesis}. However, several questions remain open. For example, the concept of independent execution of actions in two-player games have been already considered as a tool for understanding logics of imperfect information (see, e.g., \cite{Bradfield2000}, \cite{GV2012}). It is then be natural to explore the expressive power of tree games in this setting. A very general class of two-player games with concurrent behaviors, based on event structures, has recently been considered in \cite{CGW2012} where a determinacy theorem is obtained for games satisfying appropriate restrictions. Comparing our notion of tree games with their \emph{concurrent games}, and the corresponding determinacy results, looks like a promising direction for future work. Similarly, it is interesting to compare two-player non-stochastic meta-parity games with the, apparently very similar, \emph{hierarchical four player}  \emph{games} ($2$ vs $2$) of \cite{Kaiser2006}, which are used to formalize the semantics of first order logic extended with \emph{game quantifiers}. Developing verification methods is another source of possible research directions. For instance, it is interesting to study the \emph{model checking problem} for the logic pL$\mu^{\odot}$. Is it possible to compute, or at least approximate to an arbitrary degree of precision, the value $\sem{F}(p)$ assigned by a formula to a state $p$ of a finite PLTS? The problem seems far from trivial. A source of difficulty comes from the fact that (finite) two-player stochastic meta-parity games are not positionally determined. We refer to \cite[\S 8]{MioThesis} for an overview of these issues.

$\ $\\ $\ $\\
\textbf{Acknowledgments.} The author is deeply grateful to his PhD supervisor, Prof. Alex Simpson, who substantially contributed to the research presented in this manuscript with constant insights and suggestions. These thanks are extended to the anonymous reviewers for their valuable comments.
\nocite{NCI99}

\bibliographystyle{abbrv}
\bibliography{biblio}

\end{document}